\documentclass[sigconf]{acmart}
\usepackage{caption}
\usepackage{amssymb}
\usepackage{amsmath}
\usepackage{amsfonts}
\usepackage{mathtools}

\usepackage{xspace}
\usepackage{url}
\usepackage{soul, color}
\usepackage{graphicx, algorithm, color}
\usepackage{graphicx}
\usepackage{subfig}
\usepackage{bm}
\usepackage{multirow}
\usepackage{times}
\usepackage{graphics}
\soulregister\cite7
\soulregister\ref7
\usepackage{thmtools}
\usepackage{thm-restate}

\usepackage{tikz}
\usetikzlibrary{arrows.meta,positioning}
\usepackage{algorithm}
\usepackage{algpseudocode}
\usepackage{enumitem}
\usepackage[compact]{titlesec}
\titlespacing{\section}{0pt}{*1}{*}
\titlespacing{\subsection}{0pt}{*1}{*}
\titlespacing{\subsubsection}{0pt}{*1}{*}
\titlespacing{\paragraph}{0pt}{*1}{*}
\usepackage[nodisplayskipstretch]{setspace}
\def\flex{Flex\xspace}

\newcommand{\desc}[1]{}

\makeatletter
\def\BState{\State\hskip-\ALG@thistlm}
\makeatother

\def\vgg16{{VGG16}}

\def\inception3{{Inception3}}
\def\resnet50{{Resnet50}}

\usepackage{titlesec}
\titlespacing*{\section}{0pt}{0pt}{1.5pt}
\titlespacing*{\subsection}{0pt}{0pt}{1pt}
\titlespacing*{\subsubsection}{0pt}{0pt}{0.5pt}

\usepackage{etoolbox}
\makeatletter
\patchcmd{\ttlh@hang}{\parindent\z@}{\parindent\z@\leavevmode}{}{}
\patchcmd{\ttlh@hang}{\noindent}{}{}{}
\makeatother
\renewcommand\footnotetextcopyrightpermission[1]{} % removes footnote with conference information in first column
\AtBeginDocument{%
  \providecommand\BibTeX{{%
    \normalfont B\kern-0.5em{\scshape i\kern-0.25em b}\kern-0.8em\TeX}}}

\acmConference[]{}{}
\begin{document}

\title{\flex: Closing the Gaps between Usage and Allocation}


\author{Tan N. Le$^{1,2}$ and Zhenhua Liu$^1$}
\affiliation{%
	\institution{Stony Brook University$^1$, SUNY Korea$^2$}
}

\begin{abstract}
Data centers are giant factories of Internet data and services. Worldwide data centers consume energy and emit emissions more than airline industry.
Unfortunately, most of data centers are significantly underutilized.
One of the major reasons is the big gaps between the real usage and the provisioned resources because users tend to over-estimate their demand and data center operators often rely on users' requests for resource allocation.
In this paper, we first conduct an in-depth analysis of a Google cluster trace to unveil the root causes for low utilization and highlight the great potential to improve it. 
We then developed an online resource manager \flex to maximize the cluster utilization while satisfying the Quality of Service (QoS). 
%Unlike existing schedulers, our proposed scheduler leverages load estimation for scheduling.
%It also deals with estimation errors to maintain the quality of service (QoS) level.
Large-scale evaluations based on real-world traces show that \flex admits up to $1.74\times$ more requests and $1.6\times$ higher utilization compared to tradition schedulers while maintaining the QoS.
\end{abstract}

%\begin{CCSXML}
%<ccs2012>
%  <concept_id>10010520.10010553.10010562</concept_id>
%  <concept_desc>Computer systems organization~Embedded systems</concept_desc>
%  <concept_significance>500</concept_significance>
%  <concept_id>10010520.10010575.10010755</concept_id>
%  <concept_desc>Computer systems organization~Redundancy</concept_desc>
%  <concept_significance>300</concept_significance>
%  <concept_id>10010520.10010553.10010554</concept_id>
%  <concept_desc>Computer systems organization~Robotics</concept_desc>
%  <concept_significance>100</concept_significance>
%  <concept_id>10003033.10003083.10003095</concept_id>
%  <concept_desc>Networks~Network reliability</concept_desc>
%  <concept_significance>100</concept_significance>
%</ccs2012>
%\end{CCSXML}
%
%\ccsdesc[500]{Computer systems organization~Embedded systems}


\maketitle

\section{Introduction}

%% data center's energy & emission ==> urge to improve the efficiency...
Data centers are becoming the factories of Internet data and services.
There are millions of data centers in the world serving the digital demand from personal and industrial users, bringing great changes to our society. This comes at the cost of significant electricity consumption and environmental impacts. 
In total, data centers worldwide consume more than 3 percent of the world electricity and emit 2 percent of the global emissions. This is more than the airline industry \cite{dc_go_green}.
While great advancements have been made regarding data center cooling and renewalble integration, it is critical to make full use of the costly data centers.

%% Most of the time data centers are idle, why resources are wasted?
Unfortunately, data centers are inefficient in terms of both energy consumption and operation.
Many data centers do not adopt the latest technologies.
For instance, new CPUs can provide more computational power while using less electricity than previous generations.
This is important as small and medium size data centers with old technologies are responsible for 49\% of server electricity consumption in the United States \cite{delforge2014america}. 
Moreover, servers are reported to have as low utilization as 12\% to 18\% of their capacities \cite{delforge2014america}.
One of the major reasons is that users often request more than they need \cite{goolge_trace_analysis}.
One the other hand, data center operators oftentimes do not risk their business by allocating resources less than users' requests.

%% Resource usage vs. Requests.
Existing cluster resource managers like Yarn \cite{yarn} and Kubernetes \cite{kubernetes} allocate resources based on users' requests.
Since users rarely know their exact resource consumption when submitting their requests, the requested amount can be more or less than the real usage. 
In fact, our analysis of the Google cluster trace \cite{google-traces} show that the real usage is less than 45\% the requested amount on average.
Our trace analysis in \S\ref{sec:analysis} further highlights the great potential to increase cluster efficiency.

%% some related work...
Due to the utmost importance of the problem, it is not surprising to see lots of efforts have been made to mitigate the gaps between resource usage and requested amount.
In particular, oversubscription has been widely used in Yarn \cite{yarn}, Mesos \cite{mesos}, Aurora \cite{aurora}. Specifically, overscription assumes demand peaks of different users are rarely collided and thus multiplexing negatively correlated applications on the same server can accommodate more requests and therefore help improve utilization \cite{meng2010efficient}.  
Despite the benefits, server overloading may occur which results in performance degradation.
Another popular approach is to fill the cluster with low-priority jobs, which is used in resource managers such as Kubernetes and Borg \cite{kubernetes, borg,rose}. 
%Low-priority jobs have small requests but they can increase their demand when nodes are underutilized.
%The resources for these low-priority jobs are allocated in a best-effort manner.
However, using multiple priorities is not a universal solution to all data centers, and it is often chanllenging to find enough low-priority jobs to fill in the cluster. Moreover, low-priority jobs may imply low value. Therefore, it is desirable to accommodate as many high-priority jobs as possible before using low-priority jobs to increase the utilization. 

%% problem statement...
In this paper, we focus on improving the cluster utilization while maintaining the quality of service (QoS).
Specifically, given a QoS target, we developed an online resource manager to maximize the cluster utilization while satisfying the QoS target.
Our contributions are summarized as follows. 
\begin{itemize}
    \item \textbf{In-depth Google trace analysis.} We analyze $42$ GB compressed data of a 29-day Google cluster trace to figure out why the cluster is underutilized.
    At the cluster level, the total average usage is 50\% of the total capacity. 
    At the server level, the utilization level varies significantly across servers. 
    At the task level, users often overestimate their demand. 
    Since the resource manager \cite{borg} in Google Cluster relies only on the requested amount for resource allocation and scheduling, it results in load imbalance and low utilization.
    \item \textbf{Solution approach.} 
    Unlike traditional scheduling problems, we formulate a load balancing problem that does not completely rely on requests.
    It additionally takes the real-time usage of each node into consideration.   
    However, the usage of a node is uncertain and it is hard to have a perfect estimator or predictor.    
    In the solution approach, we proposed an online load balancing algorithm \flex that learns from estimation errors and automatically self-adjusts to maintain the QoS above a given target.
    \item \textbf{Google trace based evaluation}. We carry out the evaluation using a Kubernetes simulator on Google cluster trace. The evaluation shows that \flex increases the utilization by up to $1.6\times$ and admits up to $1.74\times$ more requests while maintains the QoS. \flex balances the load accross servers in the cluster and the utilization improvement is significant across a wide range of configurations. 
\end{itemize}
\section{Background \& Motivation}

\subsection{Resource Requests and Task Schedulers}

% Introduction on task schedulers.
\textbf{Task Schedulers.}
Most task schedulers rely on resource requests for scheduling.
Some of those are Yarn \cite{yarn}, Aurora \cite{aurora}, Mesos \cite{mesos}, Spark \cite{spark}, Kubernetes \cite{kubernetes}, etc. 
Based on resource requests from a task, the scheduler decides which server the task will be launched on.
Resource requests are commonly decided by users.
Users estimate or randomly pick the resource requests.
Hence, the resource requests can be more or less than the resource demands.

%Why schedulers do not rely on resource demand but resource request for scheduling. 

% What are resource requests. Resource requests vs Demand
\textbf{Request vs. Demand}.
There are several reasons that existing task schedulers do not use resource demand for scheduling.
One of the major reasons is that task demand is unknown beforehand.
People may want to pick the maximum demand as resource request. 
Since demand happens in the future, precise predictions can be very challenging.
For example, it is not easy to predict the number of users will use on a website before lauching a webserver.
Another reason is performance guarantee. 
Schedulers can guarantee resources based on request but it is hard to guarantee resources based on demand because demand changes over time and resource preemption is expensive.
Schedulers just ensure that the task receives enough resources that it requested.
%If the demand is more than request, it can be throttled.

\textbf{Resource request as maximum limit.} Some resource managers treat resource requests as maximum limit.
Users can enable Yarn \cite{yarn} and Mesos \cite{mesos} to isolate resources using Cgroups \cite{cgroups}.
For example, if an application reaches its limits, the exceeding demand of resources will be throttled.
Fair schedulers like Yarn \cite{yarn} and Mesos \cite{mesos} prefer to use this because it can guarantee fairness among applications and users.
However, it may under-utilize the nodes and it limits the performance of some applications.

\textbf{Resource request as minimum guarantee.}
To improve the utilization and performance of clusters, Kubernetes treats requests as a minimum guarantee.
Basically, an application can use more than resource requests if there are available resources in its allocated machine.
When that application uses up too much resources, it leaves no free resource for other applications.

% Serverless compute 
\textbf{Less dependent on resource requests.} 
Serverless compute platforms like AWS Lambda allow us to run tasks/applications without provisioning or managing servers. 
It reduces the overheads of server management and job configuration.
In serverless computing, it relies less on all resource requests for function-level cloud computing.
For example, the AWS Lambda framework requires only memory requests then it deicides other requests itself. 
AWS Lambda implicitly relies on resource requests that needs to deal with the same problem of using the requests.

% We should not rely on resource requests for scheduling...
In summary, We  belive that future schedulers should not rely on requests because \textit{they are not the actual demand}. 
\subsection{Google Trace Analysis}
\label{sec:analysis}

We analyze the Google cluster trace \cite{google-traces} to show that there is room for utilization improvement. We do analysis from the highest level (cluster level) to the lowest level (task level) to study the inefficiency of the Google cluster in details. 

% introduce the google trace - configuration
Google cluster trace \cite{google-traces} was collected in May 2011.
The length of the trace is 29 days.
Borg \cite{borg} is the cluster resource manager for the Google cluster.
There are around 12500 servers and 25 million tasks.
Most of tasks have both resource request and resource usage.
We do not know exactly how many CPU cores or GB RAM that tasks have used because resource usages and requests are normalized to the node that has the largest capacity.
Resource usage is randomly sampled from 5-minute usage.

% priority: ?? 

\textbf{Cluster.} Figure \ref{fig:cluster_usage} plots the total usage and request of the cluster.
The usage and request are normalized to the capacity of the whole cluster.
Both the total requests of CPU and memory are larger than the capacity sometime that means the cluster was over-subscripted.
The average total request for CPU is $1.1$ while the average total request for memory is $0.9$.
Memory is sensitive to applications so it is less over-subscripted than CPU.
The average usage of CPU ($0.43$) and memory ($0.5$) are far under from their capacities $1$. 
However, we believe that the total CPU and memory usages is not the complete proof to the inefficiency of the cluster.
So, we need to look into the details of usage and request at machine level.

\begin{figure}[h]
	\centering
	\subfloat[CPU]{\includegraphics[width=0.5\linewidth]{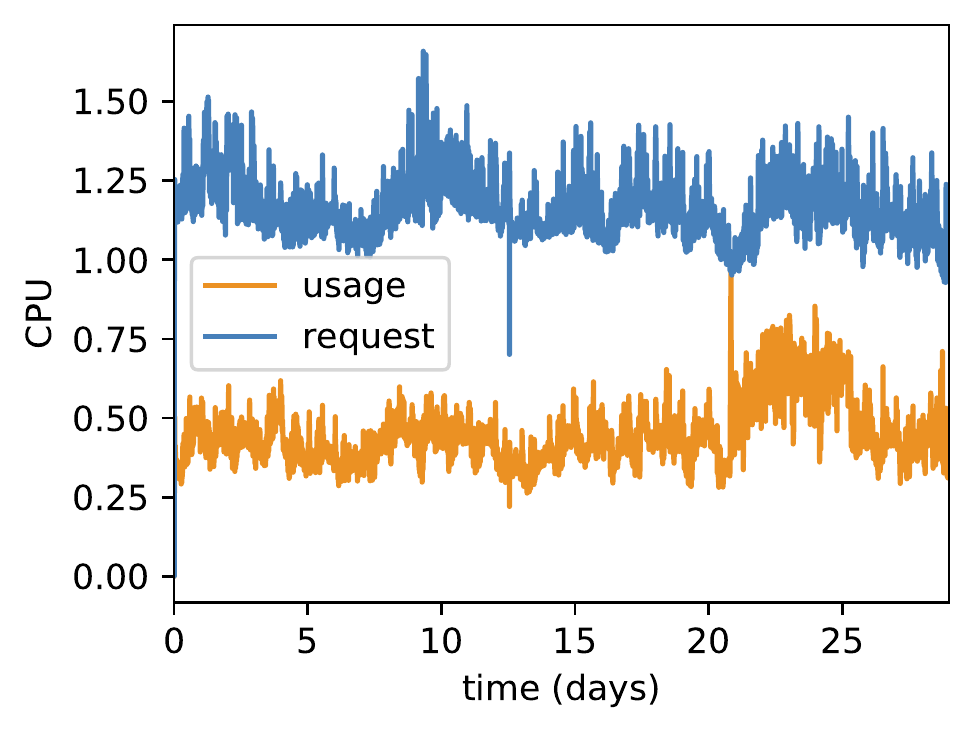}}
	\subfloat[Memory]{\includegraphics[width=0.5\linewidth]{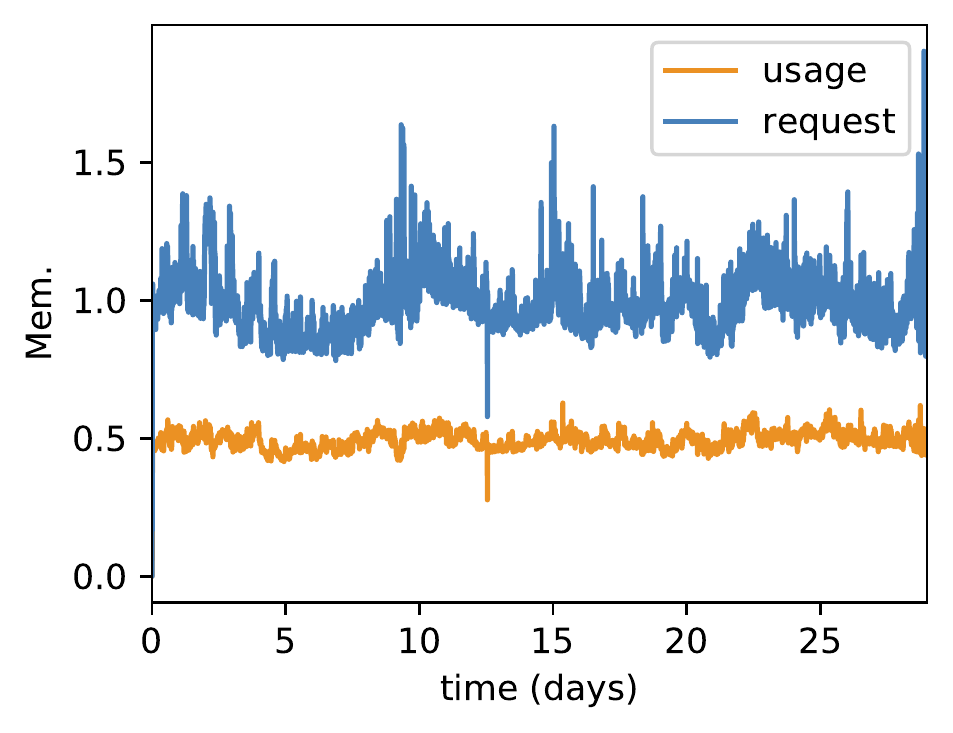}}    
	\caption{[Cluster Analysis] The total usage of cluster is highly underutilized in both CPU ($43\%$) and Memory ($50\%$).}
	\label{fig:cluster_usage}
\end{figure}

\begin{figure}[h]
	\centering
	\subfloat[Machine 1]{\includegraphics[width=0.5\linewidth]{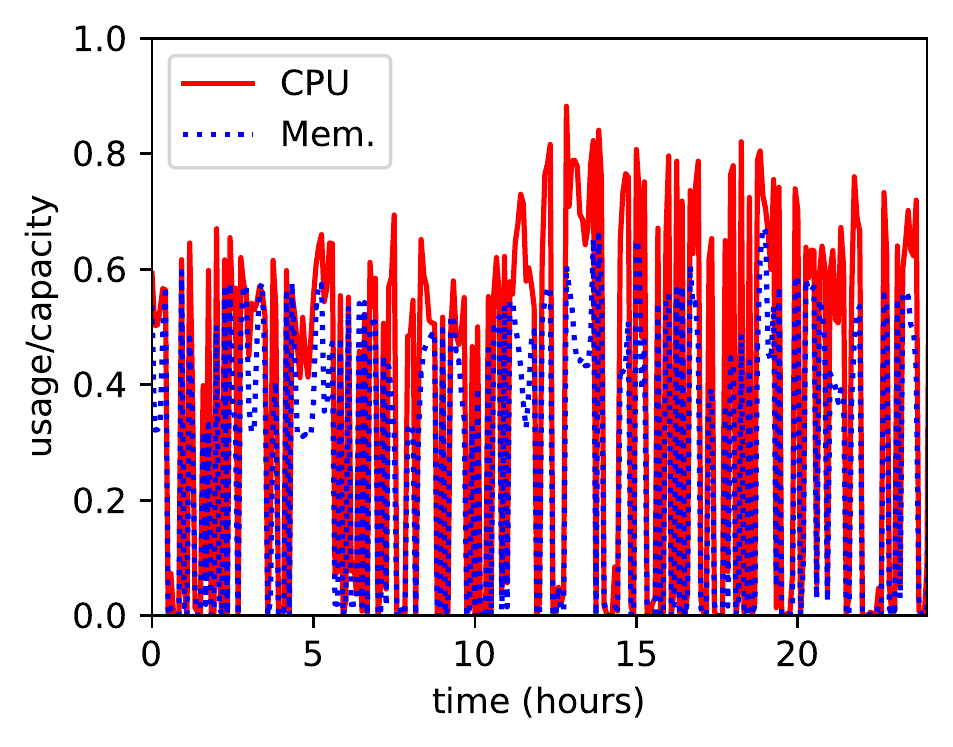}}
	\subfloat[Machine 2]{\includegraphics[width=0.5\linewidth]{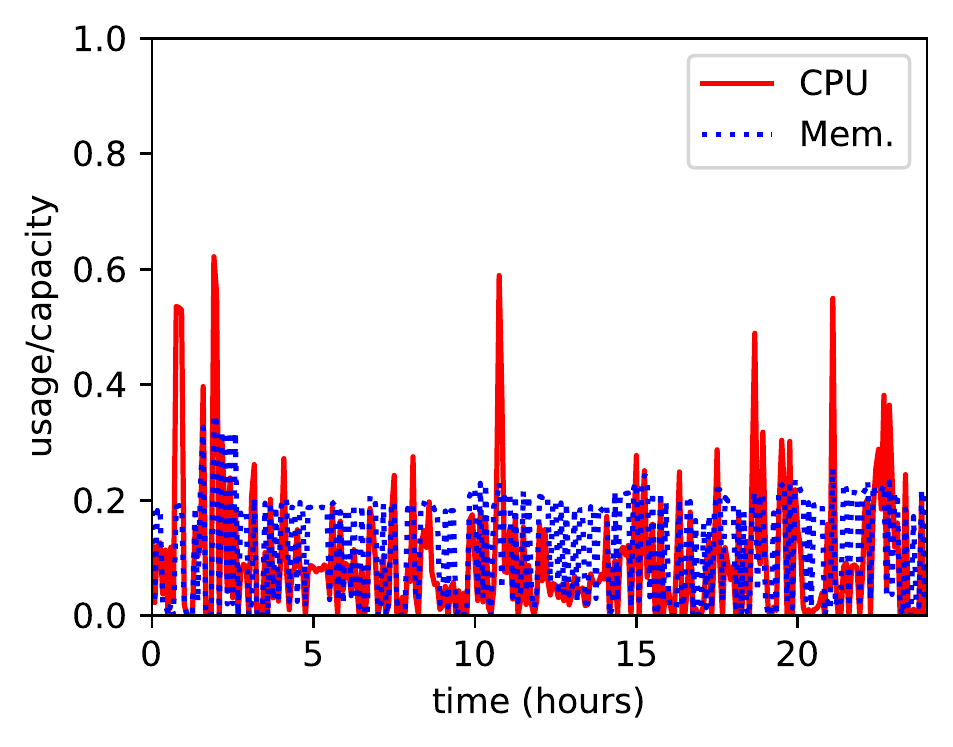}}    
	
	\subfloat[Machine 3]{\includegraphics[width=0.5\linewidth]{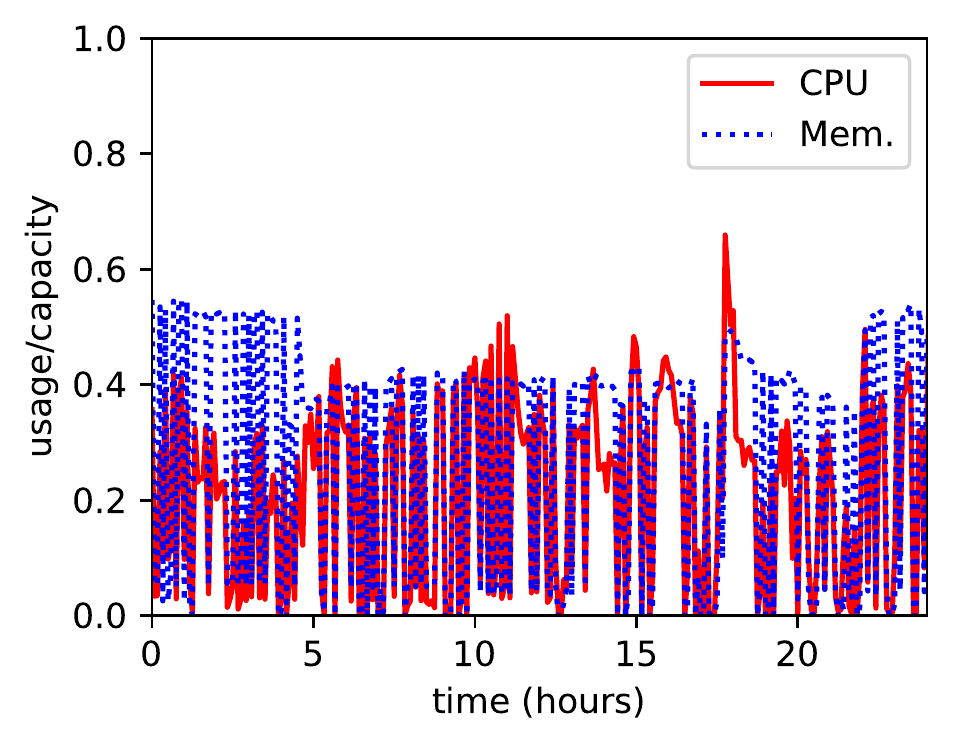}}
	\subfloat[Machine 4]{\includegraphics[width=0.5\linewidth]{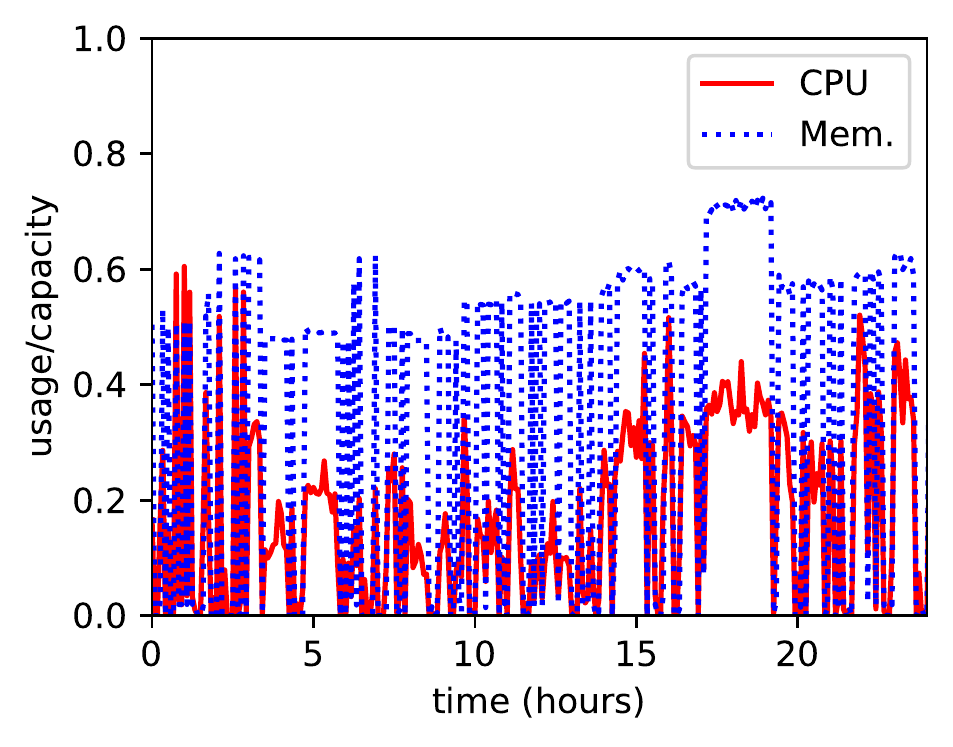}}    
	\caption{[Machine Analysis] The CPU and memory usage on different machines are not balanced.}
	\label{fig:machine_details}
\end{figure}

\textbf{Machine.} Figure \ref{fig:machine_details} plots the CPU and memory usage on four machines. 
Machine 1 has the high load while Machine 2's load is low.
This is not encouraged because the load is not evenly balanced among the servers.
Machine 4 has much greater memory usage than CPU usage but it is opposite to Machine 1.
When one of resources becomes the bottleneck on a server, we cannot allocate any more tasks to that server.

In Figure \ref{fig:machine_usage}, we plot the cumulative density function (CDF) of the usage to request ratios and the usage to capacity ratios at machine level.
Figure \ref{fig:machine_usage}a shows that more than 50\% resources are wasted in roughly 70\% of time.
More interestingly, resources are not used at all in 50\% of time.
In Figure \ref{fig:machine_usage}b, the ratio of CPU usage to request is greater than 1, meaning the usage can be actually larger than the request.
Some tasks (e.g. low priority tasks) must have increased their demand when machines have resource available.
%People may argue that servers are provisioned based on peak demand so it is reasonable .

\begin{figure}[h]
	\centering
	\subfloat[usage/capacity]{\includegraphics[width=0.5\linewidth]{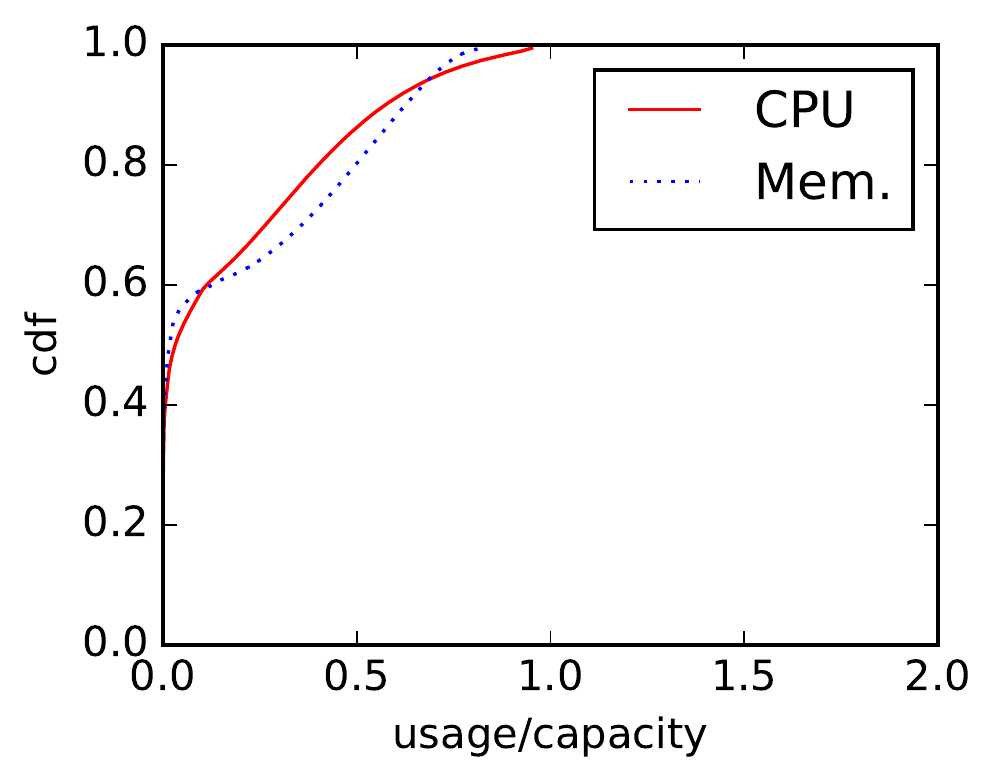}}  
	\subfloat[usage/request]{\includegraphics[width=0.5\linewidth]{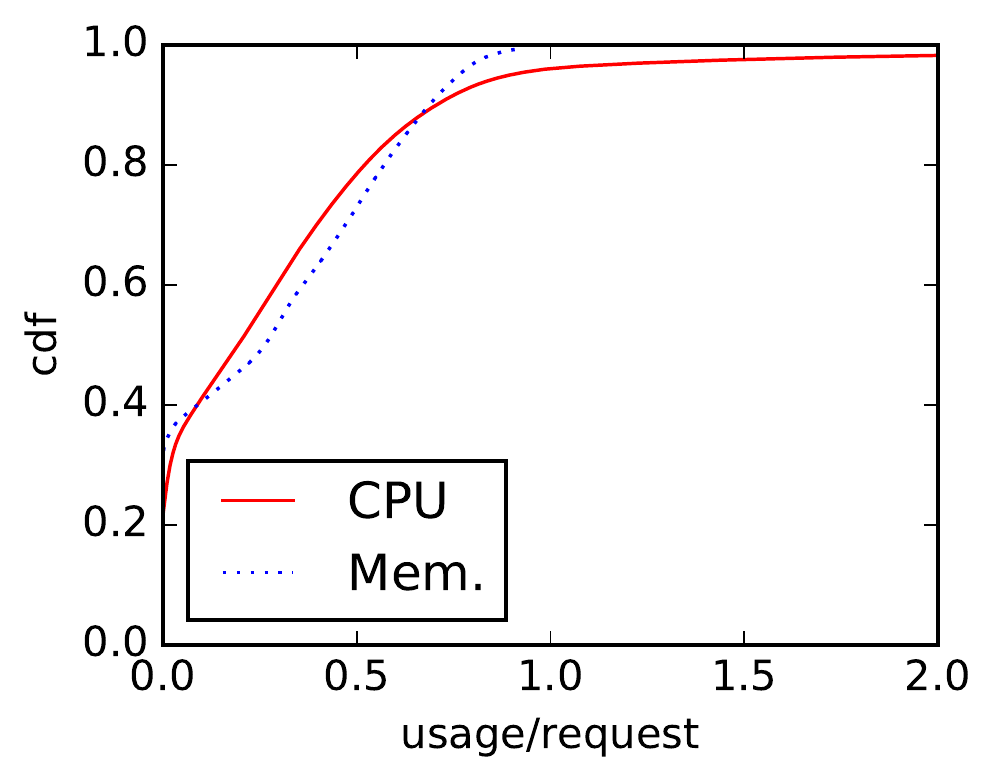}}  
	\caption{[Machine Analysis] Resource usages are less than 50\% the capacity in more than 70\% of time. Especially, resources are totally unused in 50\% of time.}
	\label{fig:machine_usage}
\end{figure}

\textbf{Task.} Figure \ref{fig:task_usage} shows the analysis on tasks.
Interestingly, the mean CPU usage can be much greater than its request while memory usages often stay under their requests (Figure \ref{fig:task_usage}a).
In the traces, the tasks with usage more than request are usually the low-priority tasks.
These tasks are submitted with small requests so they can easilly fit the remaining resources in the nodes. 
When nodes have more idle resources, Borg \cite{borg} allocates resources to the tasks in the best-effort manner so the resource usages go beyond their requests.
The peaks of resource usage can be very high compared to the resource requests (Figure \ref{fig:task_usage}b).
In Figure \ref{fig:task_usage}c, we plot the CDF of the usage standard deviation (std) of task usage normalized to its mean.
Both CPU and memory have a lot of variations.
Although maximum CPU usages are sometime much larger than requests, the standard deviation shows that memory and CPU usages have similar variations.

\begin{figure}[h]
	\centering
	\subfloat[Mean usage]{\includegraphics[width=0.47\linewidth]{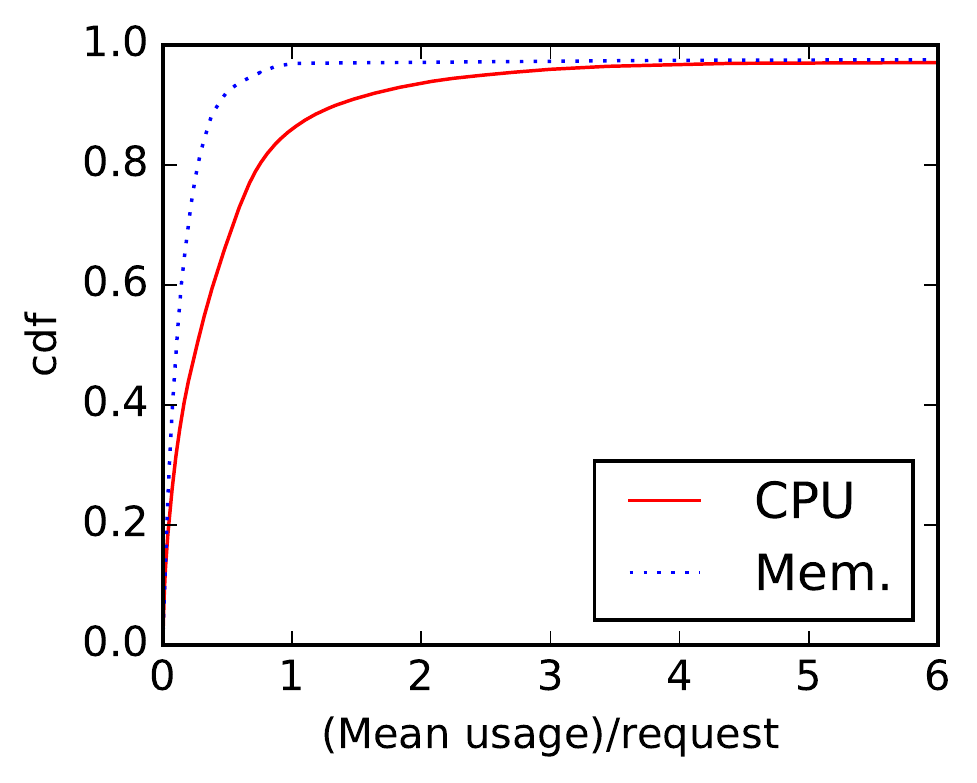}}  
	\subfloat[Max usage]{\includegraphics[width=0.47\linewidth]{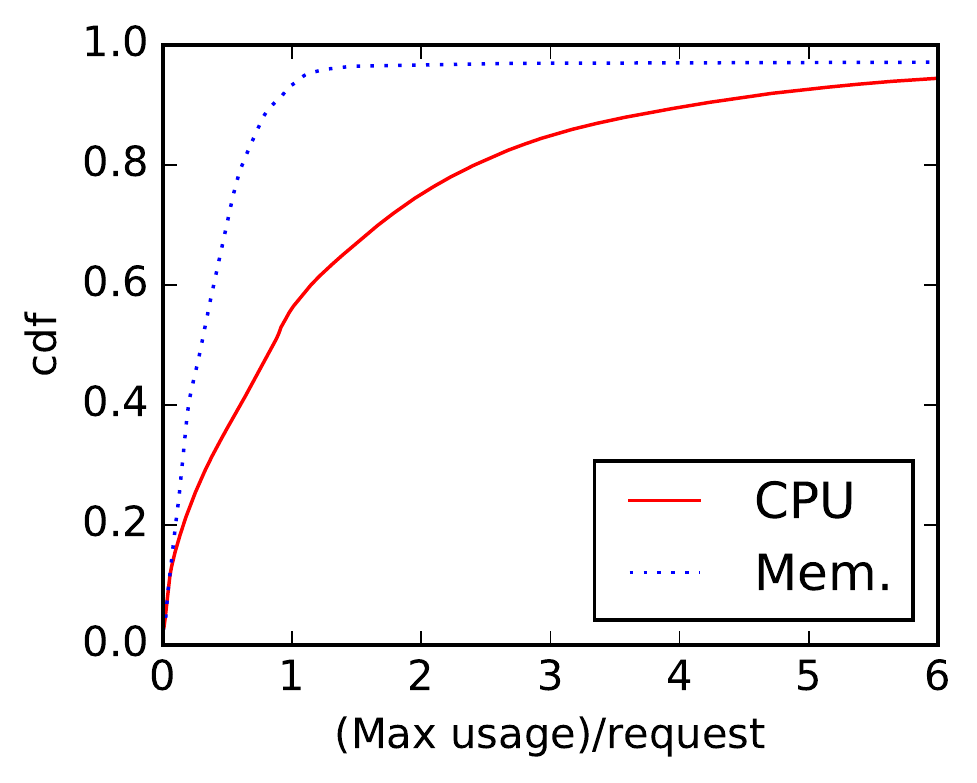}}  
	
	\subfloat[std of task usage]{\includegraphics[width=0.47\linewidth]{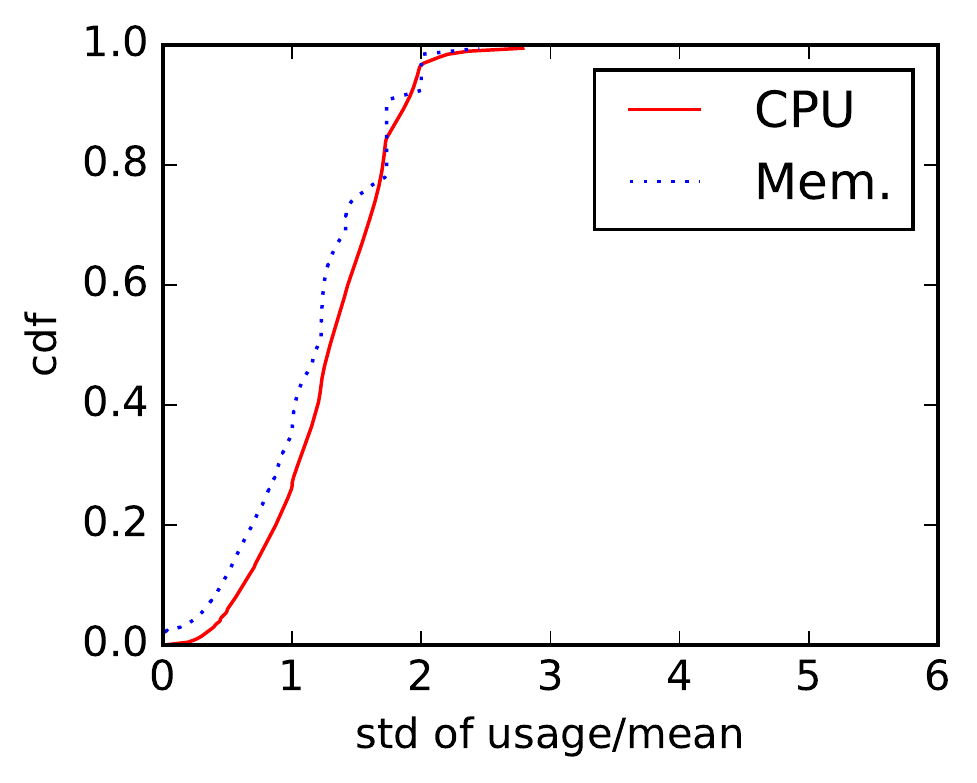}}	  
	\caption{[Task Analysis] Task usages have large variations but their means are often less than their requests.}
	\label{fig:task_usage}
\end{figure}

We classify the tasks into batch jobs, production, and system based on priorities and do analysis on resource usage in Figure \ref{fig:task_usage_p}.
The system tasks have the highest priorities.
The batch jobs have the lowest priorities.
On average, most of tasks use less resource than their requests like Figure \ref{fig:task_usage_p}(a) and (b).
In terms of maximum usage, the system tasks can use much more resource than their requests like Figure \ref{fig:task_usage_p} (c) and (d).
The maximum usages of production tasks are mostly much smaller than their requests. 
While the memory usages of batch jobs are stable, their cpu usages are more aggressive.

\begin{figure}[h]
	\centering	
	\subfloat[Mean CPU usage]{\includegraphics[width=0.47\linewidth]{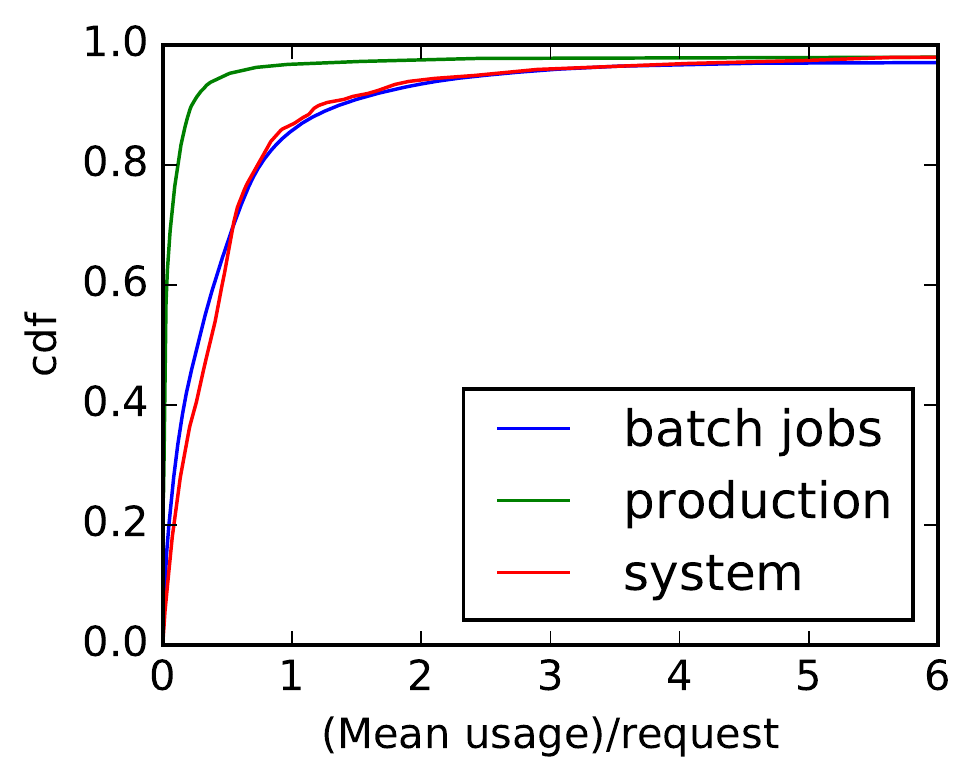}} 
	\subfloat[Mean memory usage]{\includegraphics[width=0.47\linewidth]{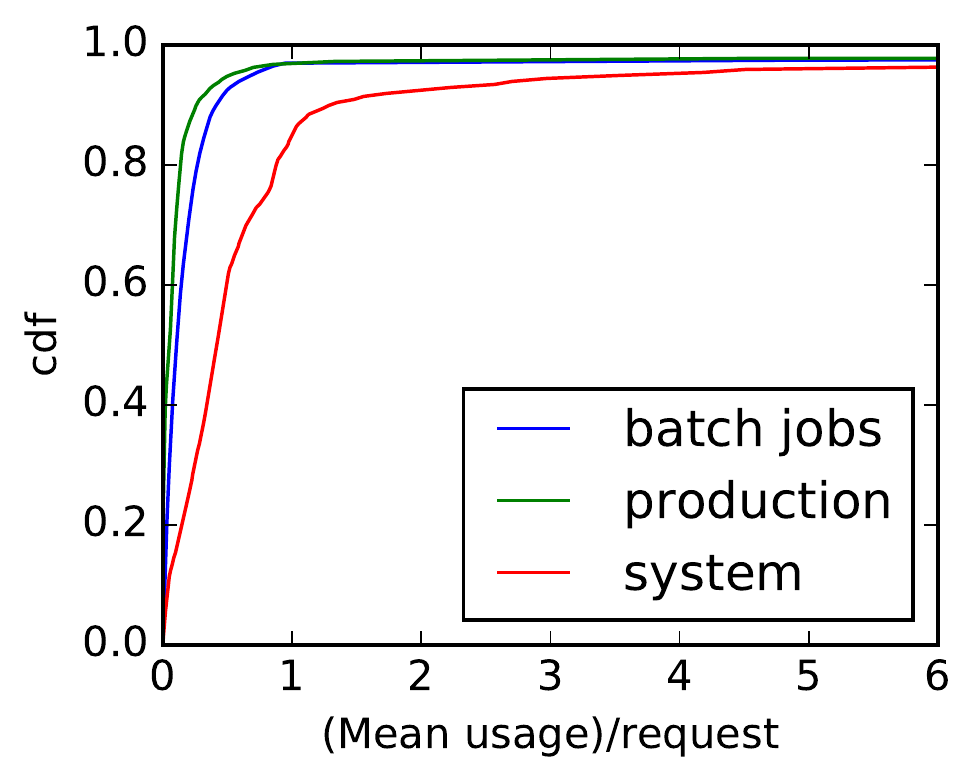}}	 
	
	\subfloat[Max CPU usage]{\includegraphics[width=0.47\linewidth]{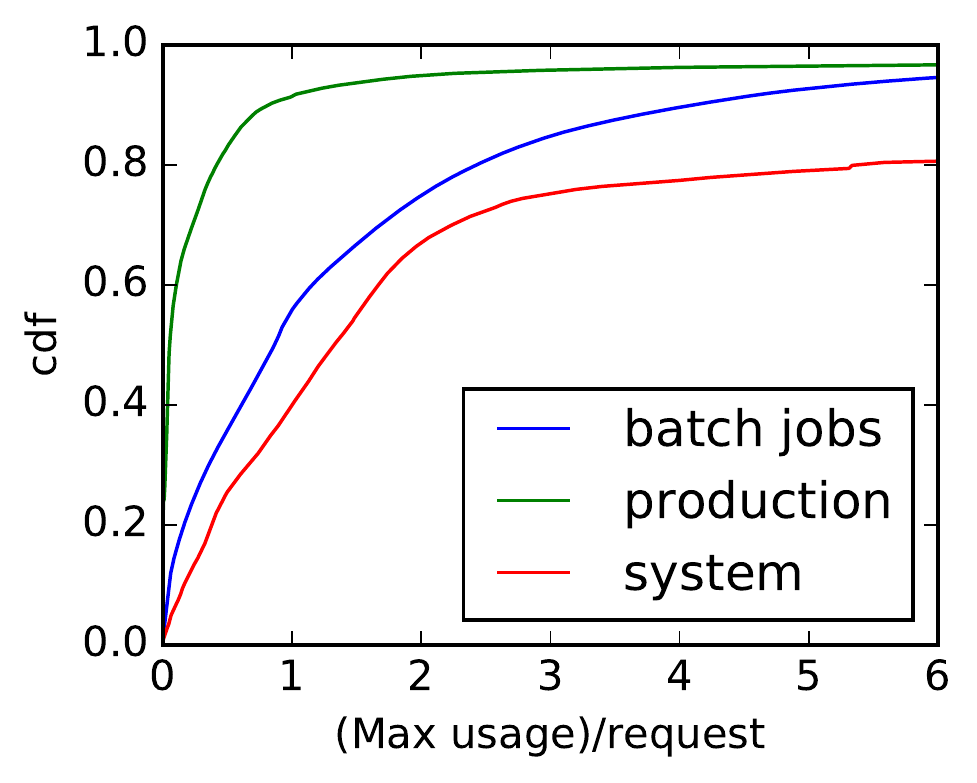}}	  
	\subfloat[Max memory usage]{\includegraphics[width=0.47\linewidth]{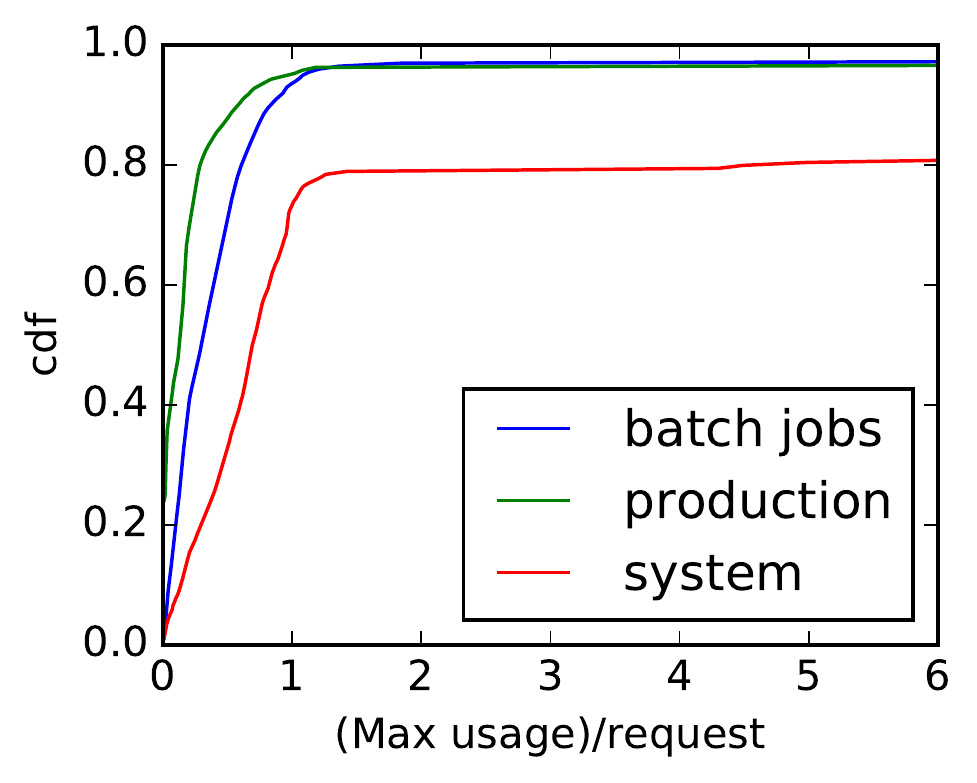}}	  
	\caption{[Task Analysis] The usage of production tasks stays close to the requests while low priority (batch jobs) and high priority (system) tasks use more resource than their requests.}
	\label{fig:task_usage_p}
\end{figure}
\subsection{Motivation}

\textbf{Request based scheduling.} Many modern schedulers like Yarn \cite{yarn}, Mesos \cite{mesos}, Aurora \cite{aurora}, and Kubernetes \cite{kubernetes} are based on requests.
They often schedule tasks using First Fit or Least Fit.
First Fit focuses on speeding up scheduling instead of improving utilization or performance.
It just picks the first node satisfying the resource constraints.
Meanwhile, Least Fit places the task on the node with the least requested resource.
Since First Fit and Least Fit naively rely the resource requests, they cannot reduce the gaps between resource usage and resource allocation.  
%It is because the total resource allocation is actually the total resource request.

From the Google trace analsis \S\ref{sec:analysis}, we could see that the resource usage is far to their capacity.
The major reason behind this is that the resource requests are often more than their resource usage.
%However, Borg \cite{borg}, the Google cluster manager, relies only on resource requests for scheduling. 

\textbf{Oversubscription.} Cluster managers have used oversubscription to admit more requests by overclaiming the cluster capacity.
Meaning, oversubscription tries to compensate the over-request.
In the Google cluster trace, the ratios of demand to request can be distinct as our analysis.  
Oversubscription causes overloads on some nodes that have major of tasks with the ratios of usage to request close to or more than $1$.
Hence, this solution cannot give much improvement.

We also observe in the Google traces \cite{google-traces} that people have used oversubsciption as well as best-effort scheduling to increase the utilization.
%However, it was not good enough.
%It is because oversubscribing resources too much may violate performance guarantee for the important applications.

\textbf{Overload detection and mitigation.} There have been several efforts on virtual machine allocation that deal with overloads \cite{baset2012towards, nathan2015towards}.
These approaches detect the hot physical nodes and try to move the virtual machines out.
They work well with virtual machines because virtual machines can be live migrated.
The migration overheads are relatively small compared to the life of virtual machines.
However, this approach is not applicable to all tasks because migration often comes with large overheads.
Schedulers like Yarn \cite{yarn} and Kubernetes \cite{kubernetes} do not support migration.
Instead, they preempt the tasks and restart them somewhere else.

\textbf{Problem statement.} \textit{Given the requests with unknown future demand, how to maximize the utilization of a clusters while maintain the QoS target?}
\section{Problem Formulation}
\label{sec:formulation}
%\todo{0.25 pages}

\begin{table}[t!]
	\begin{center}
		\caption{General notations.}
		\label{tab:table1}
		\begin{tabular}{|c|l|} 
			\hline
			\textbf{Notation} & \textbf{Description} \\
			\hline \hline
			$N$ & The set of nodes (servers) \\
			$C$ & The capacity of a node \\
			$\theta$ & Oversubsription factor \\
			$R_i$ & Requested resource on node $i$ \\
			$L_i$ & Load on node $i$ \\
			$\hat{L}_i$ & Load estimation on node $i$ \\
			$\bar{L}_i$ & Load information on node $i$ \\
			$U$ & Maximum utilization \\			
			\hline \hline 
			$J$ & The set of tasks \\			
			$\vec{d}_j$ & The vector of the resource demand from task $j$ \\
			$r_j$ & The resource request from task $j$ \\
			$s_j$ & Fair-share of remaining resource for task $j$ \\
			$x_{ij}$ & Placement decision to place task $j$ on node $i$ \\
			\hline \hline
			$q_j(t)$ & Quality of service for task $j$ at time $t$ \\
			$\rho_j$ & Quality of service target for task $j$ \\
			$Q(t)$ & Quality of service for the whole cluster at time $t$ \\
			$\rho$ & Quality of service target for the whole cluster \\
			\hline \hline 
			$P$ & Estimation penalty \\
			$\alpha$, $\beta$ & Estimation penalty update constants \\
			\hline
		\end{tabular}
	\end{center}
\end{table}

In this section, we formulate a scheduling problem for a cluster of nodes (servers).
Given the number of nodes, the task scheduler tries to admit as many tasks as possible.
To get the best performance, a scheduler often spreads the workload to all the nodes.

%\todo{State the reason why we chose load balancing.}
%The goal of this load blancing is to maintain the guaranteed performance while

% cluster
\textbf{Cluster.} There are $N$ nodes in a cluster.
We assume that all nodes have the same capacity $C$. 
It is straightforward to convert to a general problem with heterogeneous nodes.
We can add pseudo load to the small nodes until their capacities equal to the maximum one.
The total of requested resources on node $i$ is $R_i$.
$R_i$ can be greater than $C$ if the cluster is over-subscripted with the factor $\theta \geq 1$.
\begin{equation}
R_i \leq \theta C
\end{equation}
Meanwhile, the real load (usage) of node $i$ at time $t$ is $L_i$.
For the simplicity of presentation, we ignore the time $t$ in notations.
The real load $L_i$ is bounded by $C$. 
\begin{equation}
L_i\leq C
\end{equation}

% tasks
\textbf{Tasks.} At time $t$, there are $J$ pending tasks.
Task $j$ has the constant request $r_j$ and the future demand $\vec{d}_j$.
The demand $\vec{d}_j$ is unknown prior scheduling.
$\vec{d}_j$ varies from the time task $j$ scheduled till finished. 

\textbf{Resource allocation.} As the resource allocation process happens after scheduling done, we do not go too much details on resource allocation formulation.
We only consider the allocation for quality of service.
Given demand $d_j$ and request $r_j$, the resource allocation for task $j$ at time $t$ is $r_j + s_j$.
$s_j$ is commonly a fair share (FS) or weighted fair share (WFS) of the remaining available resource.
In this paper, we choose weighted fair share (WFS) as a computer often does weighted fair share for its running applications. 
$s_j$ can be negative when the real demand is less than its request. 
There are three cases.
If the total demand on a node is less than or equal to the capacity, the resource allocation of each task is equal to its demand.
If the total demand is greater than the capacity and the total request is less than the capacity, the allocator guarantees resources for all tasks based on their requests first and then splits remaining resources to all tasks using WFS.
For example, if a task requests 5 CPU cores but demands 6 CPU cores, it receives 5 CPU cores as guaranteed and extra $s_j$ CPU cores using WFS.
If the total demand and the total request are both greater than capacity, the allocator uses weighted fair share twice. 
It does WFS based on resource requests first, then does WFS again based on remained resource demand.

\textbf{Placement.} Let $x_{ij}$ be the decision variable for scheduling task $j$ on node $i$.
If the scheduler decides to place task $j$ on node $i$, $x_{ij}=1$. Otherwise $x_{ij}=0$.

\textbf{Request based load balancing (RLB)).} Existing schedulers do the load balancing based on resource requests.
It minimizes the total request on each node.
We abstract the traditional load balancing optimization as follows.

\begin{align}
	RLB:   & \min_{\mathbf{x}} R    & \\
	\text{s.t.  }
	& R \geq R_i + \sum_{j\in J} x_{ij} r_j &  \forall i \in N \\
	& R \leq \theta C  \\
	& \sum_{i \in N} x_{ij} = 1  & \forall j \in J  \\
	& x_{ij} \in \{0,1\} & \forall i \in N, \forall j \in J 
\end{align}
where $\mathbf{x}$ is the matrix representing for $x_{ij}$ for $i \in N,j \in J$.

Clearly, RLB is inefficient because it uses only resource requests.
As in our trace analysis \S\ref{sec:analysis}, there are big gaps between resource usage and resource request that causes low utilization.
If we increase the oversubscription factor $\theta$ to bridge the gaps, some of the nodes will be easily overloaded.
We suggest not completely relying on resource requests and propose a load balancing based on node usage instead. 

\textbf{Node Usage based Load Balancing (ULB).} The goal is to balance the actual load across all the nodes.
Meaning, we minimize the maximum utilization $U$.
\begin{align}
	ULB:   & \min_{\mathbf{x}} U    & \\
	\text{s.t.  }
	& U \geq \bar{L}_i +  \sum_{j\in J} x_{ij} r_j &  \forall i \in N\\
	& U \leq C  \\
	& \sum_{i \in N} x_{ij} = 1  & \forall j  \in J \\
	& x_{ij} \in \{0,1\} & \forall i \in N,j \in J 
\end{align}
where $\bar{L}_i$ is the load information on node $i$.
As the future demand is unknown, we still use the request $r_j$ in the capacity constraint $U \geq \bar{L}_i + \sum_{j \in J}x_{ij}r_j$.
We note that $\bar{L}_i$ is not the instantenous load $L_i$.
$\bar{L}_i$ should not be simply measured because the load may change overtime and the online optimization problem only captures a single snapshot. 
If $\bar{L}_i$ is too small, we can admit a lot of tasks to the same node but it will result in overloads.
If $\bar{L}_i$ is too large, it causes low utilization. 
The key question here is to where to get $\bar{L}_i$ before scheduling.

\textbf{Load Estimation Penalty.} 
In practice, load can be monitored, estimated, or predicted.
We assume that $\hat{L}_i$ is an estimated load of node $i$ at time $t$.
However, we cannot just rely $\hat{L}_i$ because we do not know how well $\hat{L}_i$ serves in $ULB$. 
To deal with underestimation or overestimation, we propose using estimation penalty $P$.
The idea is to compute $\bar{L}_i$ based on the estimated load $\hat{L}_i$ and the estimation penalty $P$ as follows.
\begin{equation}
	\bar{L}_i = P \hat{L}_i.
\end{equation}
While $\hat{L}_i$ gives us some information about the present and future load, we need to adjust the estimation penalty $P$ to avoid quality of service (QoS) violations.
If $P$ is too large, $ULB$ provides guaranteed QoS and low utilization.
If $P$ is too small, $ULB$ achieves high utilization but violates QoS.
 
\textbf{Quality of Service (QoS).} Let $q_j$ be the QoS of job $j$ at time $t$. 
$q_j=f(r_j,d_j, a_j)$ is defined based real resource usage (allocated) $a_j$ and the resource demand $d_j$ at time $t$. $q_j$ is non-decreasing on $a_j$.
Users requires $q_j$ greater or equal to the task quality target $\rho_j$.
QoS $Q$ of the system at time $t$ is computed as follows.
\begin{equation}
Q(t) = \frac{1}{|J|}\sum_{j \in J}^{} \mathbb{I}_{q_j(t)\geq \rho_j} \geq \rho
\end{equation}
where $\mathbb{I}$ is the indicator function.

% explain why it is challenging
\textbf{Challenges.} $ULB$ is an integer programming problem and well known as the NP complete problem.
It is impossible to find a optimal solution for for scheduling as it needs to be done in subsecond for thousands of nodes and millilons of tasks.
So, it requires an efficient and fast algorithm.
Furthermore, we need to pick the estimation penalty $P$ before solving $ULB$.
%We already explained $P$ cannot be too small or too large.
Since we do not have prior-knowledge of task arrival times and task demands in the future, it is challenging to pick the right estimation penalty $P$.
\section{Solution Approach}
\label{sec:approach}

We break the solution approach into 2 phases:
The first one in \S\ref{sec:loadbalancing_sol} assumes that load estimation is very accurate so we only focus on load balancing.
In the second phase, we deal with the errors from load estimation.
%The idea is to learn from the errors, and give some adjustment before doing scheduling.
We combine the two phases in the proposed algorithm.
The proposed online algorithm is compatiable with most task schedulers like Kubernetes \cite{kubernetes}, Aurora \cite{aurora}, or Yarn \cite{yarn}.

\subsection{Load Balancing with Precise Load Estimation}
\label{sec:loadbalancing_sol}

We assume that the load of a node $i$ is $L_i$ and it does not change until a new task lands in.
The load balancing now is similar to the problem of parallel machine scheduling (PMS) problem \cite{graham1979optimization}.
PMS schedules a list of jobs into multiple identical machines to minimize the makespan.
There are 3 key differences between our load balancing problem and PMS.
First, our load balancing problem allows multiple tasks to run on a machine at the same time while PMS put tasks sequentially on a machine.
Second, PMS minimizes the makespan while our problem minimizes the maximum load across the cluster.
Third, PMS knows processing times but demand is uncertain in our problem.
Since the load balancing is well known as an NP-complete problem, we are looking for an efficient online algorithm instead of optimal solutions.

When tasks arrive, they are queued up. Some schedulers like Kubernetes \cite{kubernetes} prefer to do the scheduling whenever there is any task in the queue.
%So, tasks are scheduled in a first in first out (FIFO) manner.
Other task schedulers like Yarn \cite{yarn} periodically do scheduling for the tasks on the queue.

\textbf{FIFO Scheduler.} FIFO Scheduler visits each task in a first-in-first-out manner and picks the node with lowest load for that task. 
In practice, we cannot monitor load in a real-time manner because frequent monitoring creates large overheads on the whole system.
Since users often over-request their resource demand, we use their requests $r_j$ in capacity constraints $\bar{L}_i + r_j \leq C$.

\begin{algorithm}
	\small
	\caption{FIFO Scheduler}
	\label{alg:list}
	\begin{algorithmic}[1]
		\Function{FIFOScheduler}{$J$ tasks, $N$ nodes} 
		\State $x_{ij}=0$ $\forall i,j$
		\ForAll{task $j$ in  $J$}
			\State  $\hat{i} = \arg\min_{i} L_i$ 
		\If{$\hat{L}_i + r_j \leq C$} 
			\State $x_{\hat{i}j}=1$
%			\State $L_{\hat{i}} = L_{\hat{i}} + r_{j} $
		\EndIf
		\EndFor
			\State Return $\{x_{ij}\}$ 
		\EndFunction
	\end{algorithmic}
\end{algorithm}

The computational complexity of FIFO Scheduler is $O(JN)$. Theorem \ref{theorem:list_approx} shows that the solution is $2\times$ the optimal one in the worst case.
Meaning, it requires at most $2\times$ capacity compared to the optimal solution.
%In FIFO scheduler, users may not want to over request their demand too much. 
%FIFO scheduler may never schedule large requests but it is unavoidable to any schedulers.

\begin{theorem}
	\label{theorem:list_approx}
	If the capacity $C$ is infinite, FIFO Scheduling is 2-approximation.
\end{theorem}
\begin{proof} Let $j$ is the last task with demand $d_j$ scheduled on node $i$.
Prior to scheduling, node $i$ has the lowest load $L_i-d_j$.
We have the optimal load is $L^* \geq L_i-d_j$ and $d_j \leq L^*$. Hence, $L_i - d_j + d_j \leq L^* + L^* = 2L^*$.
In the worst case, FIFO produces no more than $2\times$ the optimal solution. 
\end{proof}

\textbf{Largest Request First (LRF) Scheduler.}
The issue of FIFO Scheduler is not considering the size of pend tasks.
If tasks can be ordered, we can match the node with the lowest usage with the task with the largest demand.
Since resource demand is unknown, we use resource requests to order the tasks.
We propose the largest request first (LRF) Scheduler in Algorithm \ref{alg:heaviest_load}.
\begin{algorithm}
	\small
	\caption{Largest Request First (LRF) Scheduler}
	\label{alg:heaviest_load}
	\begin{algorithmic}[1]
		\Function{LRFScheduler}{$J$ tasks, $N$ nodes} 
			\State sort $J$ tasks such that $r_0 \geq r_1 \geq \cdots \geq r_J$.
			\State $x_{ij}=0$ $\forall i,j$
		\ForAll{job $j$ in  $J$}
			\State  $\hat{i} = \arg\min_{\mathbf{i}} L_i$ 
		\If{$\hat{L}_i + r_j \leq C$} 
			\State $x_{\hat{i}j}=1$
%			\State $L_{\hat{i}} = L_{\hat{i}} + r_{j} $
		\EndIf
		\EndFor
			\State Return $\{x_{ij}\}$ 
		\EndFunction
	\end{algorithmic}
\end{algorithm}

Algorithm \ref{alg:heaviest_load} sorts the pending tasks based on resource requests in the descending order.
Then, it finds the node with the lowest load for each task. 
By this way, the task with the largest request is sent to the node with the lowest usage.

The computational complexity of LRF is $O(Jlog(J))$ or $O(JN)$.
If there are too many pending tasks, LRF may be not efficient.
Theorem \ref{theorem:largest_approx} shows that LRF Scheduler achieves $4/3$ of the optimal solution if the order of task requests is the same as the order of task demands. Meaning, it can achieve at least $75\%$ of the optimal utilization.

\begin{theorem}
	\label{theorem:largest_approx}                               
	If the capacity $C$ is infinite and the order of resource requests is the same order of resource demands, LRF scheduling is $4/3$-approximation.
\end{theorem}
\begin{proof} Let us order $j$ tasks such that their demands $d_1 \geq d_2 \geq \cdots \geq d_j$.	
	We can assume $j$ causes the highest load $L_i+d_j$.
	If not, we keep removing $j$ from the set until we find one. 
	The removal does not change the solution of the algorithm but it decreases the optimal solution $L^{*}$.
	Since the bound holds for the assumed case, it applies to other cases with larger $L^{*}$.
	So the highest load is computed as 
	$$L_i + d_j \leq \frac{1}{N} \sum_{k \neq j}d_k + d_j \leq L^{*} + d_j.$$
	
	Now we need to prove that either $d_j \leq \frac{L^*}{3}$ or $L_i + d_j$ is the optimal solution when $d_j > \frac{L^*}{3}$.
	If $d_j \leq \frac{L^*}{3}$, the bound is clearly held.
	We use contradiction to prove the other case.
	We assume that $d_j > \frac{L^*}{3}$ and $L_i + d_j$ is not the optimal solution. 
	Assume task $l$ is the first task that makes the load on node $k$ greater than optimal, $L_k + d_l > L^*$.
	We only consider the schedule for tasks from $1$ to $l-1$.
	There is at least one task on each node.
	$d_j$ is the smallest task, and $d_j > \frac{L^*}{3}$ so there are at most two tasks on each node.
	Assume that the algorithm schedules one task to $m$ nodes  and two tasks to $N-m$ nodes.
	Clearly, let us call tasks on one-task nodes large tasks and tasks on two-task nodes small tasks.
	Task $l$ cannot be scheduled with a large task because it would create a new load greater than the optimal solution. 
	As task $l$ is the smaller than other scheduled tasks, no small or large tasks can be scheduled with a large task.
	To not violate the optimal solution, $m$ nodes schedule $m$ large tasks while  $N-m$ nodes schedule $2(N-m)+1$ small tasks.
	This is impossible because each node has at most 2 tasks.
	So, if $d_j > \frac{L^*}{3}$ and $L_i + d_j$ is the optimal solution.
	
	In summary, the solution is bound by $\frac{4}{3}L^*$.
\end{proof}

\textbf{Discussion.} 
In the case of "task by task" online scheduling, LRF Scheduler behaves like FIFO Scheduler because there is only a single task for each scheduling.
If there are multiple tasks in the queue, LRF Scheduler is expected to offer better utilization than FIFO Scheduler.
However, LRF relies on the assumption that the order of requests is consistent with the order of demands.
% In practice, this is not the case because the ratios of demand to requests vary.
% Furthermore, LRF may starve the tasks with small requests when the load is too high.

\subsection{Dealing with Estimation Errors}
\label{sec:errors}

% Impact of estimation errors.
In practice, the current load $L_i$ of node $i$ can be monitored but it is not very useful due to many two reasons.
The first reason is that load may quickly change in the future.
The second reason is that it is too expensive to monitor the load frequently. 
Since there are no perfect predictors or estimators, errors are unavoidable.
When the load is underestimated, FIFO Scheduler or LRF Scheduler admit too many tasks into a node. 
Overload can be a fatal issue.
For example, operation systems kill some of services if the memory is overloaded.
In constrast, overestimation causes low utilization.

% what to learn?
Instead of relying on the accuracy, we focus on how to deal with estimation errors.
Estimation errors are unpredictable.
An estimator works well on this workload but it may not work well on another workload.
It is hard to know which error level the system can tolerate.
The best error level at this time may not work at another time slot because the demands and systems are both dynamic.

% Why QoS not estimation errors.
Given any estimator, how to incorporate it into our solution approach?
The idea is to adjust the estimation penalty $P$.
We can pick the lower bound for P to prevent it from being too small $P>P_{min}$.
We do not adjust $P$ according to estimation errors.
It is because present errors are not very useful while the future errors are uncertain.  
Furthermore, small errors at a node do not mean that the estimator works well for another node.
Instead, we learn from the QoS of the whole cluster.
If the QoS is violated, it is a signal that the scheduler have to be less agressive.

We borrow this idea from the congestion control mechanisms in computer networks for controlling the estimation penalty $P$.
The idea of updating rule for $P$ is as follows. 
We keep monitoring the QoS of the whole system.
If QoS is violated, we quickly reset $P$ back to a larger penalty $P=P+\beta(P-P_{min})$ where $\beta$ is positive.
So, it stops the aggressiveness of admitting more tasks quickly.
If QoS is acceptable, we keep reducing the penalty $P=\alpha P$ gradually where $\alpha \in (0,1)$.   
\S\ref{sec:algorithm} presents more details how we design the updating rule for online task schedulers. 

\subsection{Online Algorithm Design}
\label{sec:algorithm}

We designed \flex algorithm for schedulers like Kubernetes \cite{kubernetes}.
When a user submit his tasks, they will be queued up in the queue.
If there is a task in the task queue, it is immediately popped out for scheduling.
If the task cannot be scheduled, Kubernetes scheduler sets the back-off time for scheduling retry.

%\todo{Multiple resources}
In practice, there are multiple resources.
For the sake of presentation simplicity, we considered only single resource scheduling in the problem formulation \S\ref{sec:formulation}.
It is straighforward to convert single resource scheduling to multiple resource scheduling.
First, contraints have to be met for each resource dimension.
For scheduling rules that cannot apply to all resource dimensions, we can pick one of the dimension instead. 
For instance, we can sort the tasks based on only memory request or the dominant resource.

We present pseudo code of \flex algorithm in Algorithm \ref{alg:proposed_algorithm}.
\textit{OnJobArrival} function processes the new arrival tasks.
\textit{ScheduleOne} is used in Kuberntes for scheduling tasks one by one from the pending queue.
The scheduler runs \textit{PeriodicEstimationPenaltyUpdate} periodically to update the estimation penalty $P$.

%\todo{Sorting Tasks, Online}
\textit{OnJobArrival} function puts a new task into the pending queue.
We can choose using FIFO or PriorityQueue.
If we use FIFO queue, the algorithm is based on the FIFO Scheduler as in Algorithm \ref{alg:list}.
If we use PriorityQueue, the algorithm works like LRF Scheduler in Algorithm \ref{alg:heaviest_load}.
Basically, we have two versions of Flex: FlexF (FIFO) and FlexL (LRF).
FlexF does not order the tasks while FlexL uses the priority queue to order the tasks from the largest memory request to the smallest request.

If there is a task pending in the queue, the scheduler invokes \textit{ScheduleOne}.
It takes the first task from the queue.
\textit{ScheduleOne} filters out all the nodes that does not meet the capacity constraint $P \hat{L}_i + r_0 \leq C$.
If there are more than one nodes remaining, it scores each node.
The rule is based on the load and the running tasks in each node.
The rule prefers the node with low load.
It also prefers the node with fewer tasks from the same source as the new task.
It is because the tasks from the same source would more likely have peaks at the same time.
To increase the speed of scheduling, we implement the filtering and scoring functions in a parallel manner.
The parrallel implementation can speed up the algorithm $p$ times.
$p$ depends on the computational power of the scheduler computing unit (e.g. CPUs or GPUs).
Finally, \textit{ScheduleOne} places the new task to the node with the highest score.

If we cannot find any node for the new task, the scheduler puts the task in the end of the queue or reports the failure error to the task owner.

Given the QoS targer $\rho$, \textit{PeriodicEstimationPenaltyUpdate} periodically updates the estimation penalty $P$.
When QoS is greater than $\rho$, it keeps reducing $P=P\alpha$ gradually where $\alpha < 1$.
The value of $\alpha$ depends on how frequent we update $P$.
If the updating period is very short, $\alpha$ has to be closer to $1$.
$P$ is lower-bound by $P_{min}$ to avoid some special cases.
For example, there are no new tasks in the queue so QoS keeps being great.
Furthermore, it prevents underestimation from happening.
We can pick $P_{min}>1$ if we expect that the estimator would not overestimate the load too much.
When QoS is less than the target $\rho$ and keeps decreasing, it is urgent to stop the aggressiveness of \flex.
There are many ways to do this.
For instance, it quickly increases: $P = P + \beta (P-1)$ where $\beta$ is a constant greater than $0$ and $P_{min}>1$.
When $\beta=1$, it doubles the overestimation, $P = 1+2(P-1)$.

\begin{algorithm}
	\caption{Flex}
	\label{alg:proposed_algorithm}
	\begin{algorithmic}[1]
		\Function{OnJobArrival}{task $j$} 
			\State Put task $j$ in queue $J$, $J$ can be FIFO or Priority Queue.
		\EndFunction
		\\
		\Function{ScheduleOne()}{} 
		\State Pick the first task $0$ in the queue $J$
		\State Filter nodes: list nodes that meets capacity: $P \hat{L}_i + r_0 \leq C$
		\If{There are more than 1 filtered nodes}
			\State Score nodes: based on load estimation and their current running tasks. 
			\State Pick node $i$ with the highest score. Dequeue task $0$
			\State Dequeue task $0$.
			\State Place task $0$ on node $i$.			
		\EndIf		
		\If{task $0$ cannot be scheduled}
			\State Dequeue task $0$ and add back to the end of the queue.
		\EndIf
		\EndFunction
		\\
		\Function{PeriodicEstimationPenaltyUpdate}{}
			\If{$Q(t) > \rho$}
				\State Reduce penalty: $P = \max(P\alpha, P_{min})$
			\ElsIf{$Q(t) < \rho$ and $Q(t) < Q(t-1)$}
			    \State Increase penalty: $P = P + \beta (P-1)$				
			\EndIf
		\EndFunction
	\end{algorithmic}	
\end{algorithm}

The computational complexity of Flex algorithm: $O(N/p)$ where $N$ is the number of nodes and $p$ is the number of parralel threads using in filtering and scoring.
\section{Evaluation}
\label{sec:evaluation}

In this section, we evaluate \flex using Google cluster trace \cite{google-traces} to show that it is better than existing modern schedulers in terms of utilization while maintains the QoS target.

\begin{figure*}[t]
	\centering
	\subfloat[Nomalized Total Request]{\includegraphics[width=0.33\linewidth]{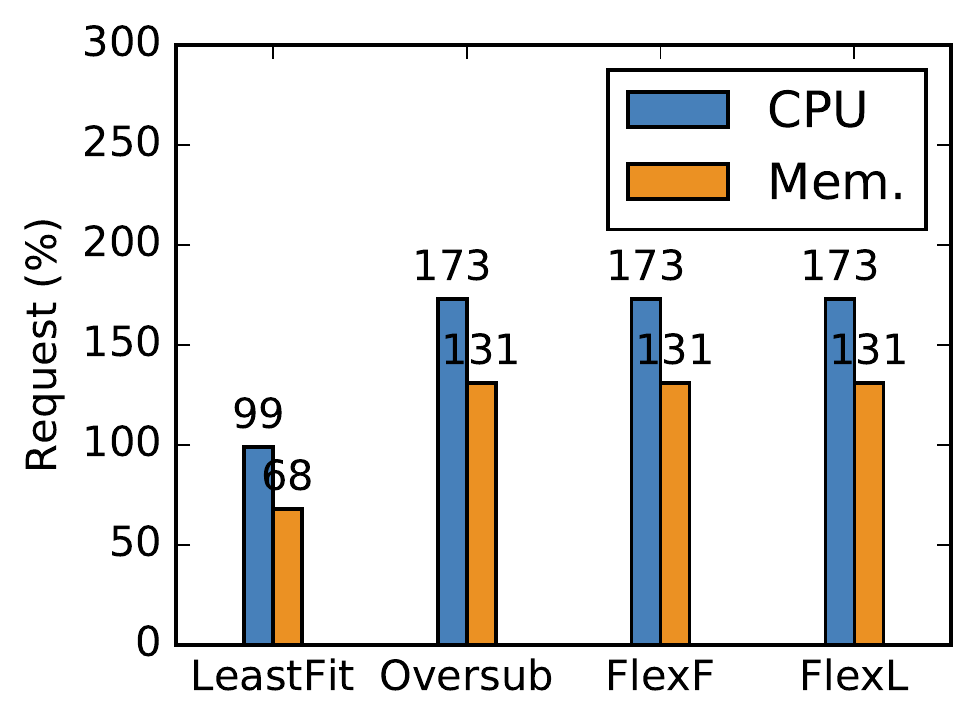} } 
	%	\subfloat[Nomalized Total Demand]{\includegraphics[width=0.33\linewidth]{figs/demand-avg} } 
	\subfloat[Nomalized Total Usage]{\includegraphics[width=0.33\linewidth]{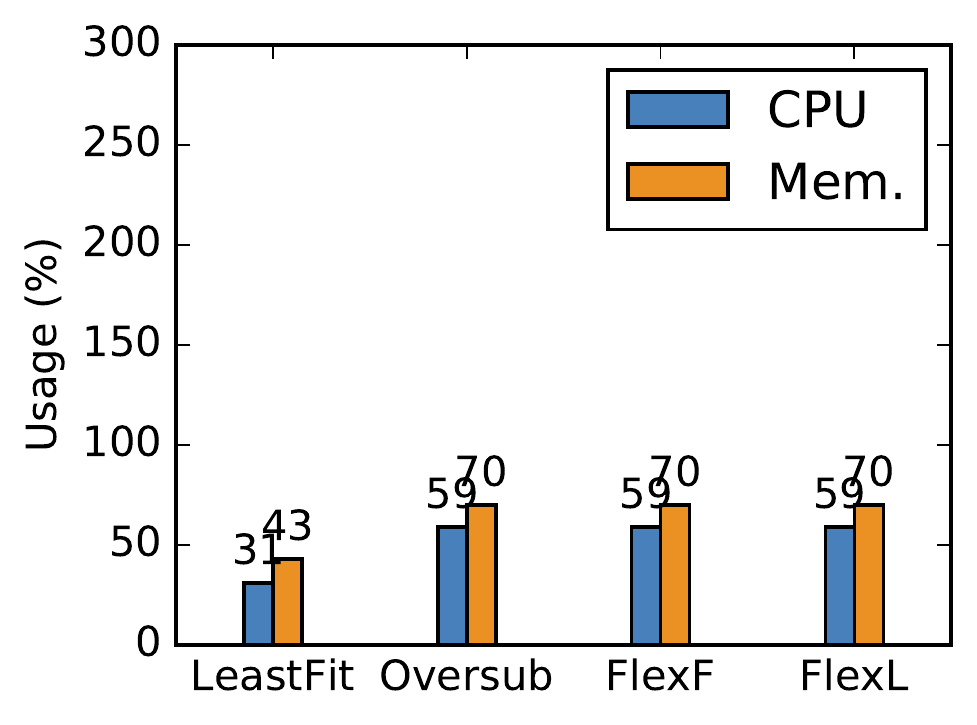} } 	
	\caption{Proposed methods \flex{F} and \flex{L} both achieve very high utilization which is similar the oversub with oversubscription factor 2. The cluster utilization using \flex{F} and \flex{L} is 70\%, $1.6\times$ of LeastFit's.}
	\label{fig:utillization}
\end{figure*}

\subsection{Setup}

\textbf{Simulator.} We extend Kubernetes Cluster Simulator \cite{k8s-cluster-simulator} that is very close to the real Kubernetes code base. The simulator was designed to test customized schedulers before trying on the real Kubernetes cluster.
Since the simulator APIs are shared with Kubernetes, developers can develop their own schedulers, plug it into the simulator for testing.
%Then, the scheduler can be easily integrated to Kubernetes.

\textbf{Configuration.} There are 4000 nodes. Each node has 64 CPU cores and 128 GB RAM. We vary the number of nodes from 3000 to 4000 in the sensitivity analysis.

\textbf{Workload trace.} We submit tasks to the simulated cluster using the task requests and usages from Google Cluster Trace \cite{google-traces}.
We submit around $714,030$ tasks within 24 hours to the simulator.
Since real demand of tasks are not available, we use their usage samples as resource demand.

\textbf{Methods.} We compare \flex with Least Fit (LeastFit) and oversubscription (Oversub).
\begin{itemize}
\item \textbf{LeastFit} allocates the job $j$ to the node $i$ in an online manner such that it minimizes the maximum of requested resources across the nodes.
LeastFit are featured in Kubernetes \cite{kubernetes} and Aurora \cite{aurora}.
\item \textbf{Oversub} combines oversubscription and LeastFit.
The oversubscription factor is $2$ and does load balancing like LeastFit does.
Most of modern resource managers like Yarn \cite{yarn} and Mesos \cite{mesos} support oversubscription.
\item \textbf{\flex{F} and \flex{L}} are the two versions of the proposed algorithm \ref{alg:proposed_algorithm}.
\flex{F} uses FIFO queue while \flex{L} uses Priority Queue to prioritize the tasks with larger memory requests first.
They start with prediction penalty $P=1.5$, $P_{min}=1$.
The estimation penalty updating constants are $\alpha=0.99$ and $\beta=1$.
\end{itemize}

\textbf{Quality of Service (QoS).} We assume that users require $a_j \geq d_j$ or $a_j \geq r_j$ for each job.
If one of aforementioned conditions satisfied at time $t$, $q_j(t)=1$.
The QoS of the cluster is $Q(t) \geq 0.99$. 

\textbf{Estimator.} We use a simple load estimator in this evaluation.
We assume that the future demand does not change much.
Meaning, we monitor and use the current resource usage for scheduling.
We show that our proposed algorithm works well even with this simple estimator.

\textbf{Metrics.} We use resource utilization and QoS as the main metrics for the evaluation.

\subsection{Resource Utilization}

To evaluate the resource utilization of four methods, we compute the request and usage of the whole cluster.
They are normalized to the total cluster capacity.
While request represents for the amount of workload the cluster handles, usage is the actual utilization of the cluster.
Figure \ref{fig:utillization} plots the request and usage of the four methods.
\flex{F} and \flex{L} admit the most requests and are similar to the Oversub with oversubscripting $2\times$ of the cluster capacity.
Since they admit more $74\%$ than LeastFit, \flex{F}, and \flex{L} handle more demand and increase the utilization of cluster up to $1.6\times$.
%The next question is how QoS changes when we over-admit requests.

\subsection{Quality of Service (QoS)}

\begin{figure*}[!t]
	\centering
	\subfloat[Cdf of quality of service (QoS)]{\includegraphics[width=0.33\linewidth]{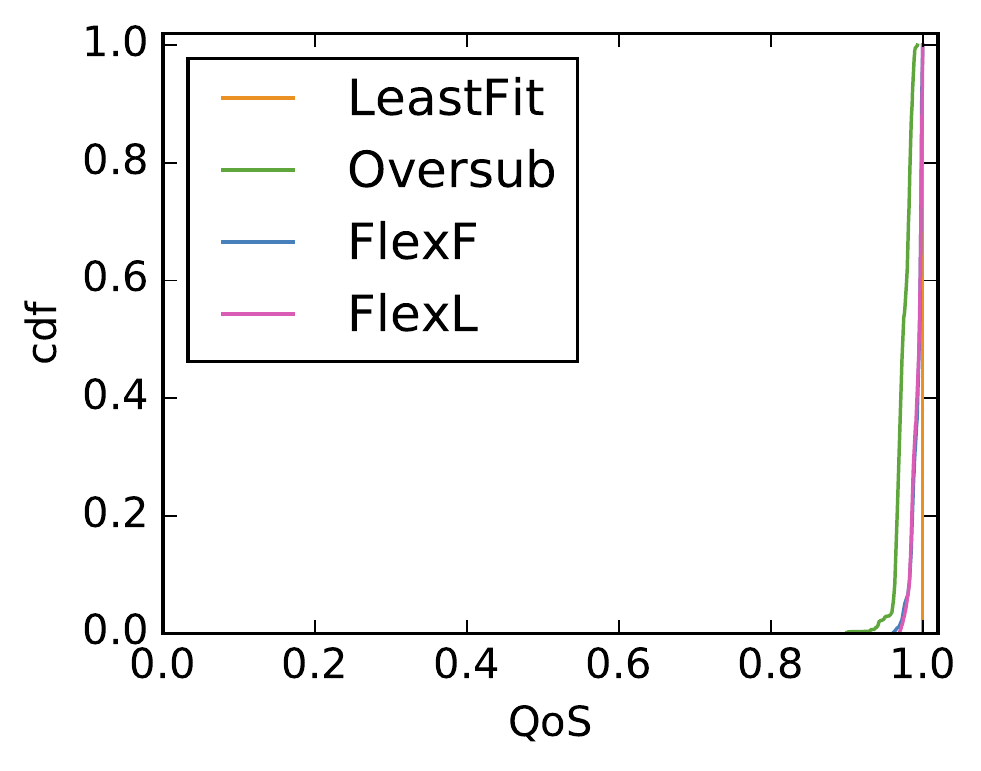} } 	
	\subfloat[QoS violations]{\includegraphics[width=0.33\linewidth]{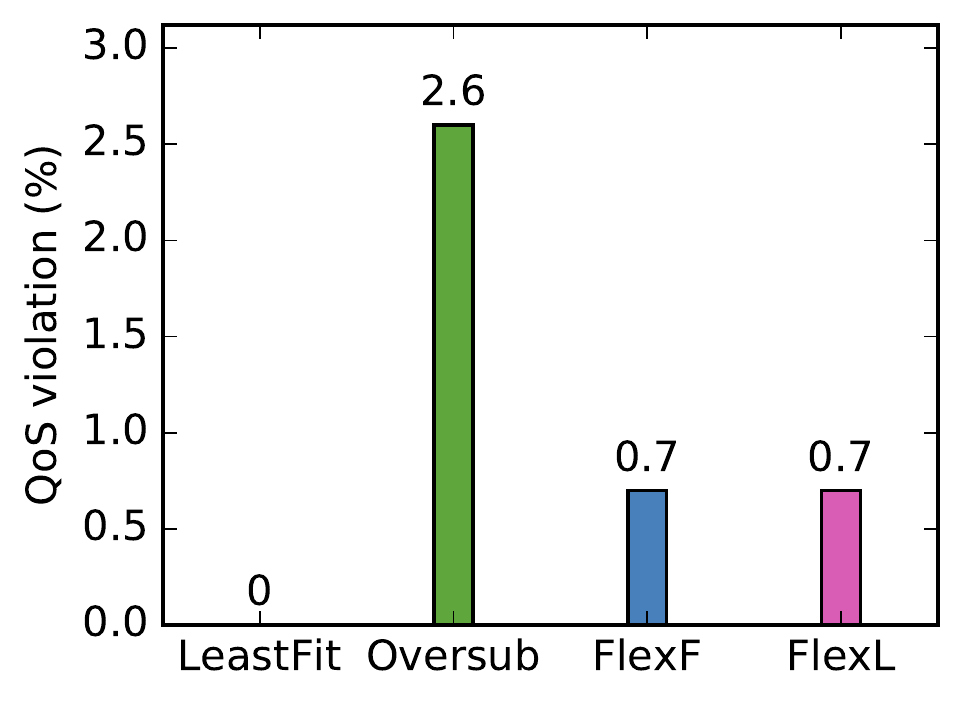} } 
	\caption{Both proposed methods maintain 99\% QoS gurantee while Oversub violates performance guarantee.}
	\label{fig:qos}
\end{figure*}

To evaluate the Quality of Service (QoS), we plot the cumulative density function (cdf) of QoS over time and the percentage of QoS violations in Figure \ref{fig:qos}.
In Figure \ref{fig:qos}a, \flex{F} and \flex{L} are better than Oversub although they have the same utilization.
Clearly, LeastFit has the least utilization but it is the best in terms of QoS because it does not over-admit requests.
Figure \ref{fig:qos}b shows the average percentage of QoS violations.
%Oversub is the worst with QoS violation is $2.5 \%$ far from the target $1\%$.
Both \flex{F} and \flex{L} are $3.7\times$ better than Oversub.

\subsection{Estimation Penalty}

To understand why \flex{F} and \flex{L} are better than Oversub in terms of QoS, we plot the QoS and the estimation penalties over time in Figure \ref{fig:penalty}.
The QoS of Oversub suffers from large drops but cannot quickly recover.
Meanwhile, the QoS of \flex{F} and \flex{L} also have some drops but they quickly recover to maintain the QoS target.
The changes of estimation penalty explain this.
For example, when there is a QoS is less than target (99\%) at around 1 hour, the estimation penalty immediately goes up to stop the aggressiveness of admitting more requests.
It stops making QoS worse and waits until QoS is better.

\begin{figure}[!htb]
	\centering
	\subfloat[Qualify of service (QoS)]{\includegraphics[width=0.7\linewidth]{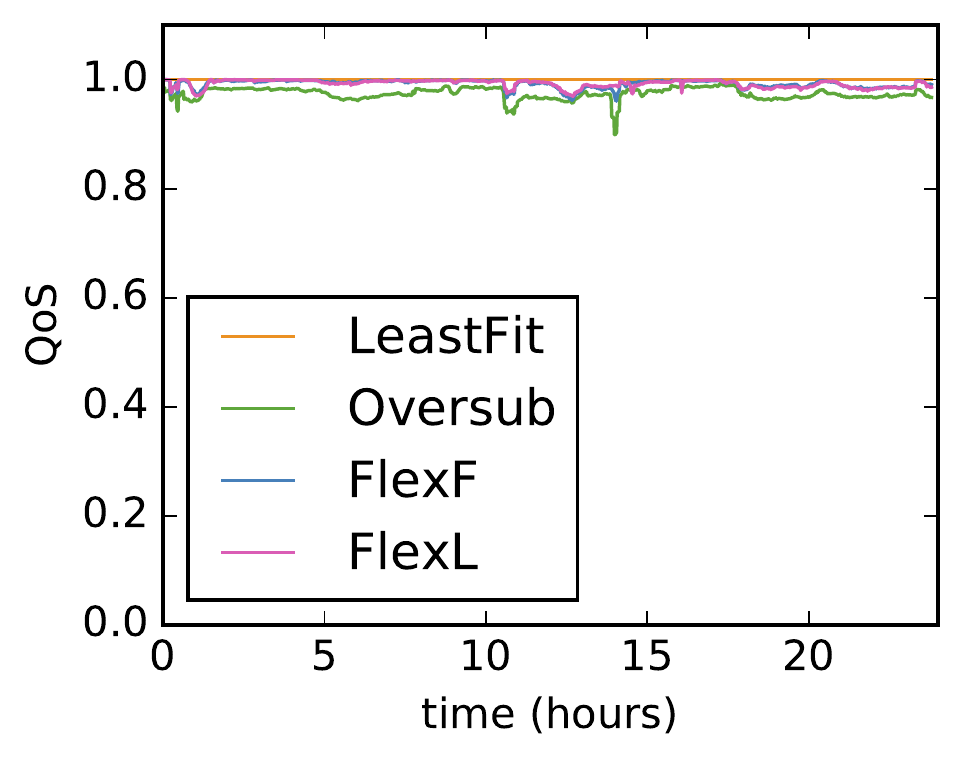} } 
	
	\subfloat[penalty]{\includegraphics[width=0.7\linewidth]{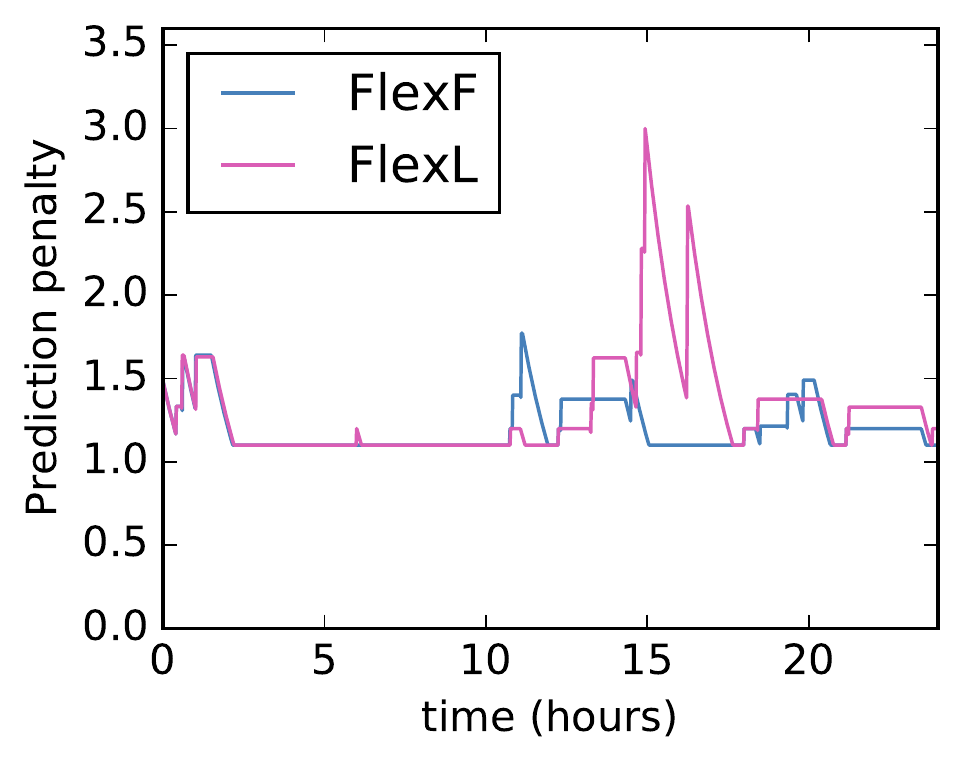} } 
	\caption{The estimation penalties of \flex{F} and \flex{L} immediately react to the QoS degradation.}
	\label{fig:penalty}
\end{figure}

\subsection{Load Balancing}

To evaluate how well \flex{F} and \flex{L} do load balancing, we compute the standard deviation (std) of memory usage across nodes over time in Figure \ref{fig:stdmemusage}.
The standard deviation is normalized to the usage mean.
Small standard deviations mean the load are spread well to all nodes.
The standard deviations \flex{F} and \flex{L} are less than standard deviations of LeastFit and Oversub.
Meaning, \flex{F} and \flex{L} maintains the best load balance among the four methods.
Oversub is the worst because it admits more requests into a node than LeastFit.

\begin{figure}[!htb]
	\centering
	\includegraphics[width=0.7\linewidth]{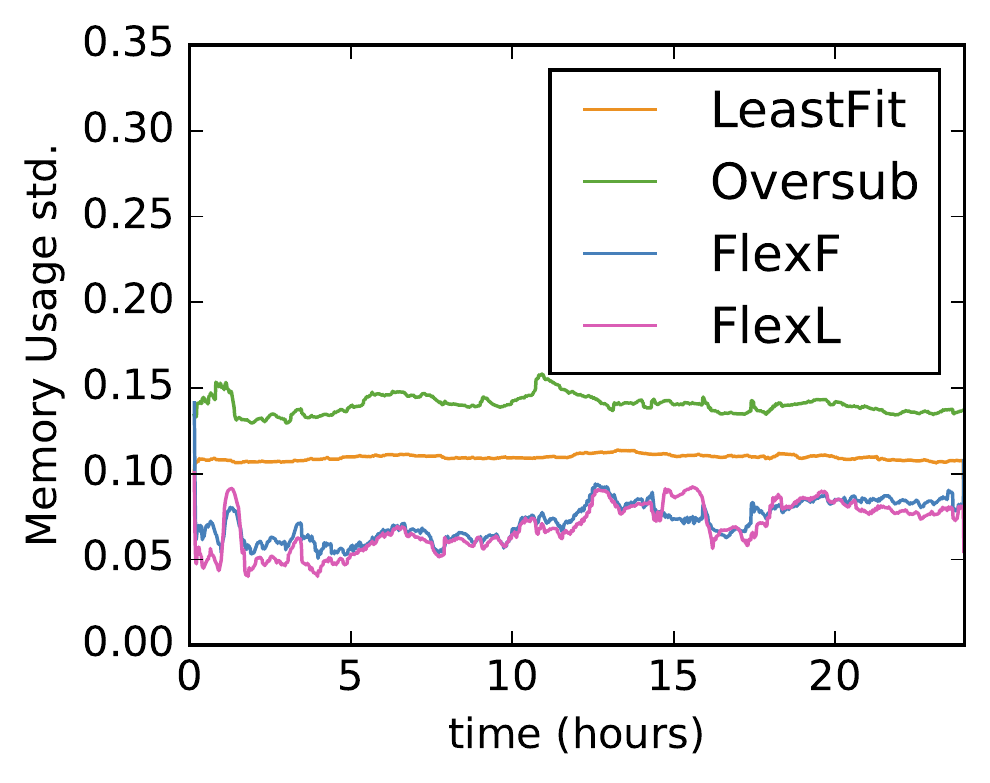}
	\caption{\flex{F} and \flex{L} are better than LeastFit and Oversub in terms of load balancing.}
	\label{fig:stdmemusage}
\end{figure}

\subsection{Sensitivity Analysis}
In this section, we carry out the sensitivity analysis on cluster size and the ratios of demand to request, 

\begin{figure}[!htb]
	\centering
	\subfloat[Average total request]{\includegraphics[width=0.7\linewidth]{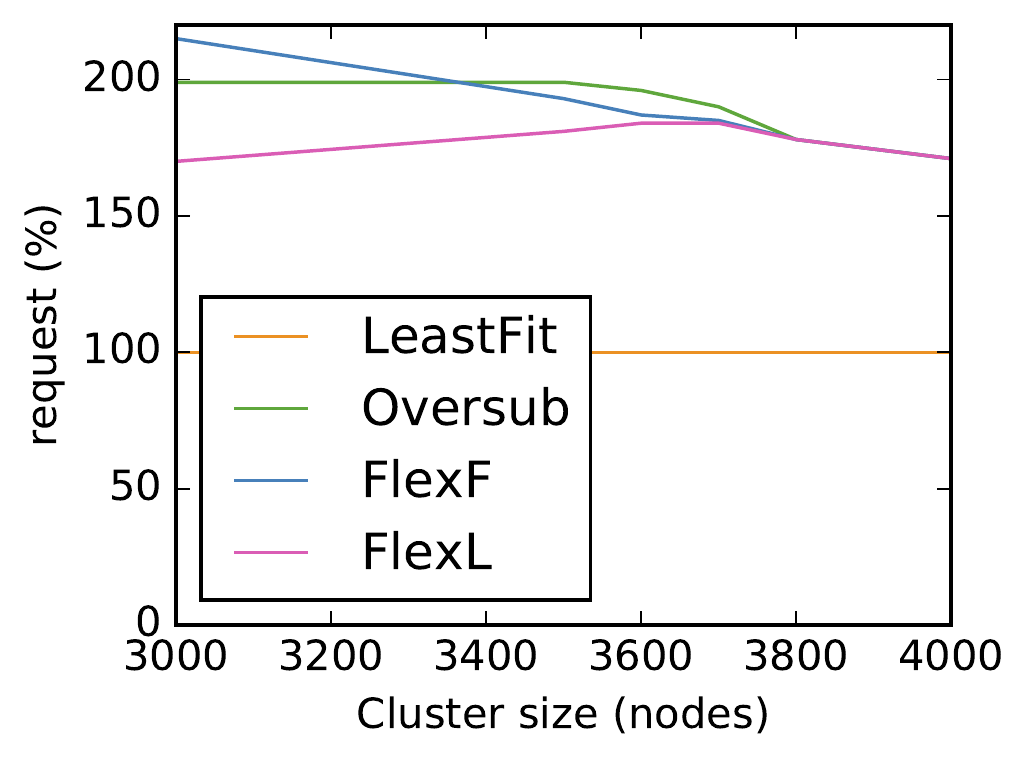} } 
	
	\subfloat[Quality of service (QoS)]{\includegraphics[width=0.7\linewidth]{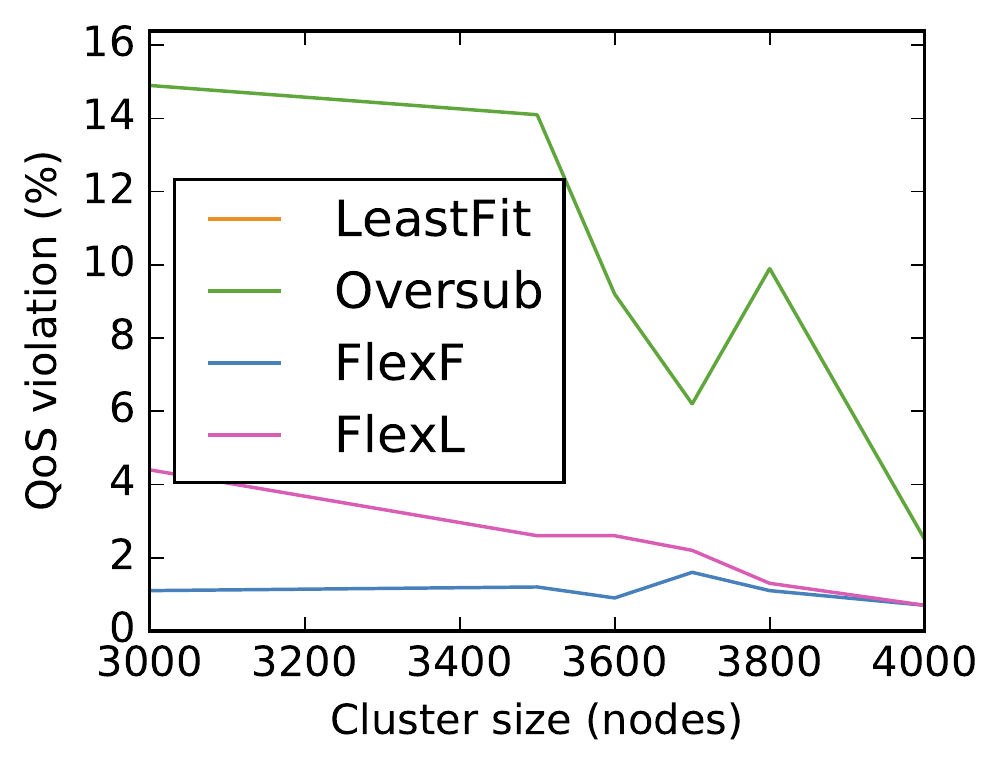}} 
	\caption{[Sensitivity] In terms of QoS, both \flex{F} and \flex{L} significantly outperform Oversub at the small cluster size (3000) while it gives a similar utilization to Oversub.}
	\label{fig:clustersize}
\end{figure}

\textbf{Cluster size.} To study the impact of low load or high load on utilization and QoS, we vary the cluster size from 3000 to 4000 nodes.
%When the cluster is small (3000), Oversub cannot handle the same amount workload as much as the large cluster (4000).
Figure \ref{fig:clustersize} shows that the utilization of both \flex{F} and \flex{L} work well and similar to Oversub.
In terms of QoS, Oversub is the one suffering the most at the small cluster size (3000).
Under high load, Oversub tries to pack more tasks into nodes than it can handle so overloads more often.
Both proposed algorithms do not suffer much from the change of cluster size. 
\flex{F} is better than \flex{L} because \flex{L} prefer tasks with larger requests which have larger variations in demand.

\textbf{Ratios of demand to request.} 
To study the impact of ratios of demand to request on performance, we scale up and down the demand but do not change the request in Figure \ref{fig:ratios}.
The goal is to simulate the cases that users estimate their requests differently.
Since Oversub does not consider the real demand, it suffers the most from the QoS violation when we scale up the demand (1.5).
In contrast, \flex{F} and \flex{L} are less aggressive on admitting more requests when the demand scale is high.
Meaning, \flex{F} and \flex{L} are less dependent on the ratios of demand to request.
Hence, \flex{F} and \flex{L} still maintain the small percentages of QoS violation.

\begin{figure}[!htb]
	\centering
	\subfloat[Average total request]{\includegraphics[width=0.7\linewidth]{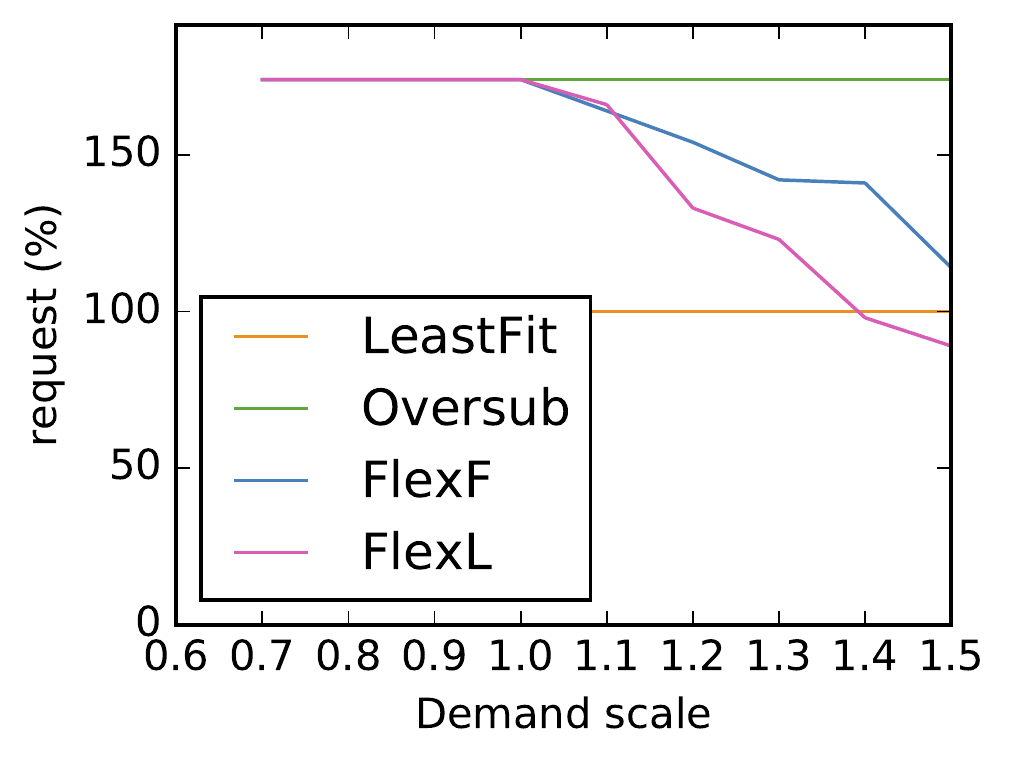} } 
	
	\subfloat[Qualify of service (QoS)]{\includegraphics[width=0.7\linewidth]{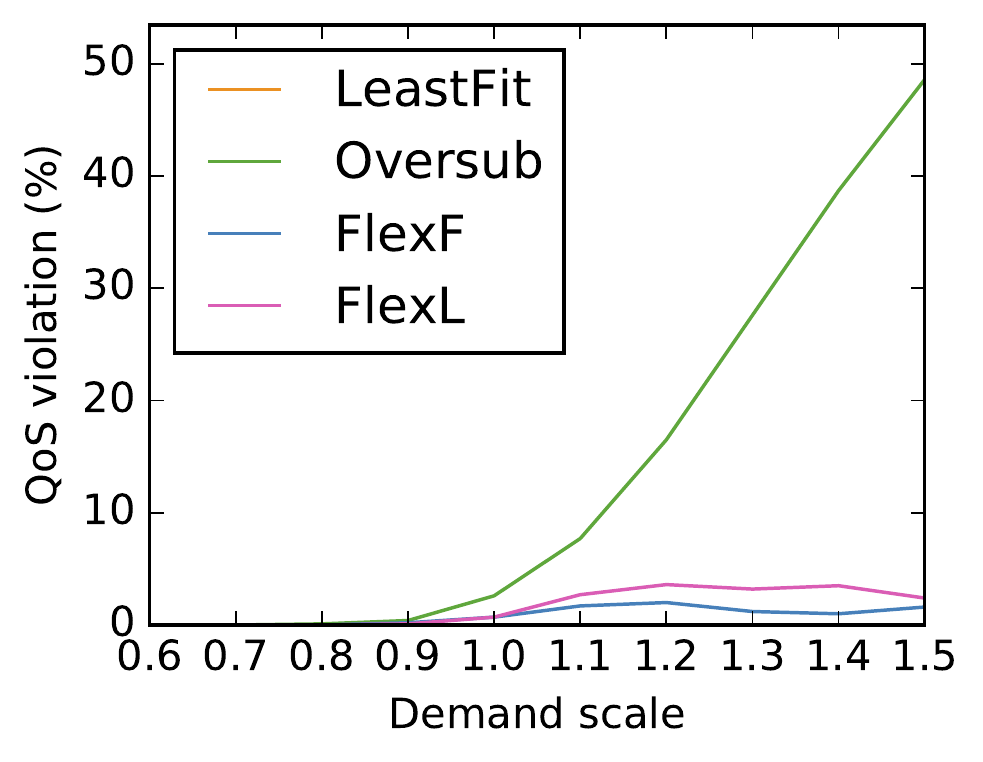}} 
	\caption{[Sensitivity] The large demand to request ratios significantly degrade the QoS of Oversub but not \flex{F} and \flex{L}.}
	\label{fig:ratios}
\end{figure}
\section{Related Work}

\textbf{Request based Schedulers.} Existing schedulers in Yarn \cite{yarn}, Mesos \cite{mesos}, and Kubernetes \cite{kubernetes} rely on resource requests for scheduling.
As in the Google trace \cite{google-traces}, requests are often overestimated that leads to low utilization.
Supprisingly, most of modern schedulers like DRF \cite{drf}, Carbyne \cite{carbyne}, HUG \cite{hug}, a nd BoPF \cite{bopf} still rely on requests.
Some other systems like AlloX \cite{allox} estimate the requested resource themselves.
However, estimating resource usage is still very challenging to apply to all types of workloads.

\textbf{Multiple priorities.} Data centers classify workload into multiple priorities.
The goal is to not only guarantee resources but also improve the utilization of the whole cluster.
Schedulers treat the low priority and long tasks like analytics in a best-effort manner.
For example, Kubernetes \cite{kubernetes} or Borg \cite{borg} can allocate more resource than requests to tasks when the nodes have free resources.
To utilize the idle resource, Rose \cite{rose} consolidates the low utilization machines with low priority workload like speculative tasks.

\textbf{Oversubscription} was proposed to deal with over-estimate use of resources \cite{tomas2014autonomic, williams2011overdriver, breitgand2012improving}.  Clusters can take more virtual resources than the physical resources. 
These approaches assume that applications can tolerate some degrees of oversubscription so a node can be allocated with more requests.
Oversubscription may cause the overload on some nodes and fail to provide performance guarantee. 

\textbf{Multiplexing.} Multiple VMs can be multiplexed in the same physical machine when their usage peaks do not collide \cite{meng2010efficient,tomas2013improving, chen2011effective, nandi2012stochastic}. 
The idea is to pick the best set of VMs with negative correlation of resource usages to be placed in the same physical machine.
However, this requires statistic knowledge of the VM resource usage which is commonly unknown.
It is even harder for acquiring the resource usage of computing tasks ahead.

% mix VMs, detect and mitigate.
\textbf{Overload detection.} There are many techniques dealing with overloads in virtual machine allocation.
Since the resource usage of VMs are uncertain, overload detection and mitigation are commonly used in VM oversubscription \cite{baset2012towards, nathan2015towards}.
Detecting overload is based on monitoring and predicting.
Wood proposed Black-box detects hotspot by using thresholds and Gray-box predict hotspots by linear regression \cite{wood2007black}.
Similarly other prediction methods are used like exponentially weighted moving average (EWMA) \cite{xiao2012dynamic}.
CloudScale \cite{cloudscale} assumes the resource usages of applications in a VM is known so it can predict the overload can happen.

\textbf{Overload mitigation.}
Overload mitigation can be done by shutting down VMs or live migration. 
Shutting down VMs is expensive for old VMs while live migration is more acceptable because VMs live a long time \cite{clark2005live, travostino2006seamless}.
The live migration overhead could be large if they use a lot of disk and memory. 
If the VMs use network drives, the overhead is much smaller. 
Although there have been several efforts on overload mitigation, it is still not applicable for all workloads such as short tasks.
\section{Conclusion}

In this paper, we developed an online resource manager \flex to maximize the cluster utilization while maintaining the QoS.
Our Google trace analysis shows that the Google cluster has low utilization although the cluster is already oversubcripted. 
Instead of relying on users' requests, we formulate an online scheduling problem based on load estimation and developed \flex that combines both load balancing and feedback control.
%The goal of designing \flex is for real online task schedulers like Kubernetes scheduler.
%We prove that \flex fast and efficient.
Google trace driven evaluations show that \flex achieves significantly higher resource utilization compared to conventional schedulers while maintaining the QoS.
\section{Acknowledgements}
This research is supported by NSF grants CNS-1730128, 1919752, 1617698, 1717588 and was partially funded by MSIT, Korea, IITP-2019-2011-1-00783.
\label{EndOfPaper}

%\begin{acks}
%\input{scripts/ack}
%\end{acks}

\bibliographystyle{include/ACM-Reference-Format}
\bibliography{include/refs}

%\appendix
%\input{scripts/appendices}

%\newpage
%\section*{Structure}
%
%\begin{enumerate}
%    \item Introduction: 1.5
%    \item Background: current practice, data analysis to show the potential: 1.5
%    \item Problem formulation: 1
%    \item Solution approach: 2
%    \item Performance evaluation: 3
%    \item Related work \& conclusion: 1
%\end{enumerate}

\end{document}